\documentclass[a4paper,10pt]{article}

\usepackage[utf8]{inputenc}
\usepackage{amsthm}
\usepackage{color}
\usepackage{amsthm}
\usepackage{graphicx}
\usepackage{amsmath}
\usepackage{amsthm}
\usepackage{amssymb}
\usepackage{color}
\usepackage{geometry}
\usepackage{graphicx}
\usepackage{url}

\newcommand{\blue}{\textcolor{black}}

\newcommand{\bound}{k^2 - 6}
\newcommand{\boundcap}{K^2 - 6}

\newcommand{\cI}{\mathcal{I}}
\newcommand{\cC}{\mathcal{C}}
\newcommand{\cT}{\mathcal{T}}

\newcommand{\ab}{\underline{ab}}

\newcommand{\steven}{\textcolor{black}}

\newcommand{\NN}{\mathbb{N}}

\newcommand{\cR}{\mathcal{R}}

\newcommand{\lca}{\text{lca}}
\newcommand{\umast}{\text{umast}}
\newcommand{\n}{\{1,2,\ldots ,n\}}

\newtheorem{lemma}{Lemma}
\newtheorem{theorem}{Theorem}
\newtheorem{corollary}{Corollary}

\newtheorem{observation}{Observation}

\title{Satisfying ternary permutation constraints by multiple linear orders or phylogenetic trees}

\author{Leo van Iersel, Steven Kelk, Nela Leki\'c, Simone Linz}

\begin{document}
\maketitle

\begin{abstract}
A ternary permutation constraint satisfaction problem \steven{(CSP)} is specified by a subset $\Pi$ of the symmetric group $S_3$. An instance of such a problem consists of a set of variables $V$ and a set of constraints $\cC$, where each constraint is an ordered triple of distinct elements from $V$. The goal is to construct a linear order $\alpha$ on $V$ such that, for each constraint $(a,b,c) \in \cC$, the ordering of $a,b,c$ induced by $\alpha$ is in $\Pi$. Excluding symmetries and trivial cases there are 11  such problems, and their complexity is well known. Here we consider the variant of the problem, denoted 2-$\Pi$, where we are allowed to construct two linear orders $\alpha$ and $\beta$ and each constraint needs to be satisfied by at least one of the two. We give a full complexity classification of all 11 2-$\Pi$ problems, observing that
in the switch from one to two linear orders the complexity landscape changes quite abruptly and that hardness proofs become rather intricate. We then focus on one of the 11 problems in particular,
which is closely related to the $2$-{\sc Caterpillar Compatibility} problem in the phylogenetics literature. We show that this particular CSP remains hard on three linear orders, and also in the biologically
relevant case when we swap three linear orders for three \steven{phylogenetic} trees, yielding the $3$-{\sc Tree Compatibility} problem. \steven{Due to the biological relevance of this problem we also give extremal results concerning the minimum number of trees required, in the worst case, to satisfy a set of rooted triplet constraints on $n$ leaf labels.}
%Our empirical experiments suggest that this function grows extremely slowly.
\end{abstract}

\section{Introduction}
\steven{A \emph{ternary permutation} constraint satisfaction problem (CSP), sometimes also known as an \emph{ordering} CSP}, is specified by a subset $\Pi$ of the symmetric group $S_3$. An instance of such a problem consists of a set of variables $V$ and a set of constraints $\cC$, where each constraint is an ordered triple of distinct elements from $V$. The goal is to construct a linear order $\alpha$ on $V$ such that, for each constraint $(a,b,c) \in \cC$, the ordering of $a,b,c$ induced by $\alpha$ is in $\Pi$. For example, if $\Pi = \{ 123, 132 \}$ then for each constraint $(a,b,c)$ we require that $\alpha(a) < \alpha(b) < \alpha(c)$ or $\alpha(a) < \alpha(c) < \alpha(b)$,
which can be summarized as $\alpha(a) < \min(\alpha(b), \alpha(c))$. Excluding symmetries and trivial cases there are 11 such problems, some of which have acquired specific names in the literature, such as \emph{betweenness} \cite{chor1998geometric} and \emph{cyclic ordering} \cite{galil1977cyclic}. A full complexity classification of the 11 problems is given in \cite{guttmann2006variations} and summarized in Table \ref{tab:summary}. Due to their fundamental character
these problems have stimulated quite some interest from the approximation \cite{guruswami2011beating}, parameterized complexity \cite{gutin2012every} and algebra \cite{bodirsky2010complexity} communities.

In this article we consider the variant of the problem, denoted $k$-$\Pi$, where we are allowed to construct $k$ linear orders
%$\alpha$ and $\beta$
and each constraint needs to be satisfied by at least one of them.
%To the best of our knowledge this problem has not appeared earlier in the literature, although %as we shall discuss in due course it has a very natural application in computational biology.
We give a full complexity classification of all 11 2-$\Pi$ problems, observing that in the switch from one to two linear orders the complexity landscape does not behave monotonically (see Table \ref{tab:summary}). We note that for a single linear order the polynomial-time variants can be solved with variations of topological sorting, and the hard variants can be proven NP-complete using fairly straightforward reductions \cite{guttmann2006variations}. In the case of two linear orders all the polynomial-time variants are \emph{trivially} solveable while, for the other variants, establishing the NP-completeness 
is much more challenging, requiring a wide array of novel gadgets and constructions.

Following this classification, we shift our focus to \emph{phylogenetics}, a branch of computational biology concerned with the inferrence of evolutionary histories \cite{SempleSteel2003}. Here we are given a set of \emph{rooted triplets} which are leaf labeled, rooted binary trees on three leaves.
% \steven({see Figure X)}.
\steven{Rooted triplets have a central and recurring role within phylogenetics due to the fact that they can be viewed as the atomic building blocks of larger evolutionary histories \cite{SempleSteel2003}}. The goal in the $k$-{\sc Tree Compatibility} problem, \steven{introduced in \cite{linz2013optimizing}} is to partition the set of triplets into at most $k$ blocks such that the triplets inside each block can be topologically embedded into a tree-like hypothesis of evolution known as a phylogenetic tree. This problem is particularly topical given the growing awareness that a genome often contains multiple tree-like 
evolutionary signals that need to be untangled \cite{HusonRuppScornavacca10,morrison2011introduction,nakhleh2013computational}. Determining whether a single block (i.e. a single tree) is sufficient can be computed in polynomial-time, using the classical algorithm of Aho \cite{AhoEtAl1981}. More recently, Linz et al. \cite{linz2013optimizing} determined that it is NP-complete to determine whether two trees are sufficient. We observe that the problem 2-$\Pi_1$ (which corresponds to the
$\Pi = \{ 123, 132 \}$ example given earlier) is equivalent to the problem Linz et al. studied, with one \steven{extra} restriction: the two phylogenetic trees we construct must be ``caterpillars'', %that is, linear
%orders,
yielding the $2$-{\sc Caterpillar Compatibility} problem. Our hardness result for 2-$\Pi_1$ thus \steven{supplements} the hardness result of Linz et al. (and, as a spin-off result, establishes a link to the literature on the dichromatic number problem \cite{bokal2004circular}). To further extend this result we show that 3-$\Pi_1$
%(the problem of deciding whether or not there exist three linear orderings, say $\alpha$, %$\beta$, and $\gamma$, such that, for each constraint $(a,b,c)$, the ordering of $a$, $b$, %and $c$ that is induced by at least one of $\alpha$, $\beta$, and $\gamma$ is in $\Pi_1$)
 is hard and, building on this machinery,
$3$-{\sc Tree Compatibility} is also hard.
As with many of the hardness results in this article we make heavy use of special gadgets that are unique solutions
to the set of constraints that they {\blue imply}. Such uniqueness gadgets are likely to be of independent interest in their own right.

We then explore $k$-{\sc Tree Compatibility} from an extremal perspective. In particular, 
how many blocks (i.e. phylogenetic trees) are required, in the worst case, to satisfy a set of triplets on $n$ leaf labels?
%if we are given a set of rooted triplets on $n$ leaf labels, how many  blocks (i.e. phylogenetic %trees) are required, in the worst case, by optimum solutions to the tree compatibility problem? %This question is inspired by the results of
Empirical experiments
%, also described here, which
show that this function $\tau(n)$ grows
%\emph{extremely} slowly. In fact, $\tau(n)$ grows
so slowly that the question arises: does there exist a constant $c$ such that $c$ trees are always sufficient? We prove that the answer is \emph{no}, showing
that $\tau(n) \rightarrow \infty$ as $n \rightarrow \infty$. This shows that the
$k$-{\sc Tree Compatibility} problem does not become trivial for any $k> 1$, which supports
the conjecture of Linz et al. that the problem is NP-complete for every $k>1$.
We also show a logarithmic upper bound.

We conclude with some open problems and future directions for research.

%; the 4 problems that can be decided
%in polynomial time can be solved with variations of topological sorting. 

%A ternary Permutation-CSP is specified by a subset $\Pi$ of the symmetric group $S_3$. An instance of such a problem consists of a set of variables (or taxa) $X$ and a multiset of constraints, which are %ordered triples of distinct variables of $X$. The objective is to answer weather all constraints given by an instance of $\Pi$ can be satisfied by a number of linear orders. In this paper we will look at two %linear orders. We use 2-$\Pi_i$ notation to refer to a problem in which we ask wether constraints in $\Pi_i$ can be satisfied by two linear orders.

\begin{table}
\begin{center}

\begin{tabular}{l l c c}
	 &				&1LO	&2LO \\ \hline
 $\Pi_0 \text{ (linear ordering)}$ & 123 				& P 	& NPC \\
 $\Pi_1$ & 123, 132 			& P 	& NPC \\
 $\Pi_2$ & 123, 213, 231 		& P 	& P \\
 $\Pi_3$ & 123, 231, 312, 321 		& P 	& P \\
 $\Pi_4$ & 123, 231			& NPC 	& NPC \\
 $\Pi_5 \text{ (betweenness)}$ & 123, 321			& NPC 	& NPC \\
 $\Pi_6$ & 123, 132, 231		& NPC 	& NPC \\
 $\Pi_7 \text{ (circular ordering)}$ & 123, 231, 312 		& NPC 	& P \\
 $\Pi_8$ & $S_3 \setminus$ 123, 231  	& NPC 	& P \\
 $\Pi_9 \text{ (non-betweenness)} $ & $S_3 \setminus$ 123, 321  	& NPC 	& NPC \\
 $\Pi_{10}$ & $S_3 \setminus$ 123 	& NPC 	& P 
\end{tabular} 
\caption{The complexity of the 11 ternary permutations CSPs, in the case of 1 linear order (1LO) (see \cite{guttmann2006variations}) and 2 linear orders (2LO) (this article). }
\label{tab:summary}
\end{center}
\end{table}

%\newpage

\section{Preliminaries}
This section provides notation and terminology that is used in the remainder of the paper. Preliminaries in the context of
phylogenetics are given in the second part of this section.

\subsection{Ternary permutation constraint satisfaction problems}
%So fox example, problem 2-$\Pi_5$ consists of constraints of form (123, 321). To say that a linear order satisfies a $\Pi_5$-constraint $c=(abc)$ is to say that in that linear order either triple $abc$ or triple $cba$ is satisfied. \textcolor{red}{So a constraint is a set of triples and we use brackets to denote a constraint $(abc)$ and nothing to denote a triple $abc$}

In this paper, we investigate eleven ternary permutation constraint satisfaction problems that are all based on subsets of the symmetric group on three elements, i.e. $$S_3 = \{123, 132, 213, 231, 312, 321\}.$$ An instance $\cI=(V,\cC)$ of such a problem consists of a set $V$ of {\it variables} and a set $\cC$ of {\it constraints}, where each constraint is an ordered triple $(v_1,v_2,v_3)$ of three distinct variables of $V$.

Let $\alpha:V\rightarrow\{1,2,\ldots,m\}$ be a linear ordering of $V$, where $|V|=m$. We sometimes write $$(\alpha^{-1}(1),\alpha^{-1}(2),\ldots,\alpha^{-1}(m))$$ to denote $\alpha$. 

Let $V$ and $V'$ be two sets of variables such that $V\cap V'=\emptyset.$ Furthermore, let $\alpha$ be a linear ordering of $V$, and let $\beta$ be a linear ordering of $V'$. We use $\bar{\alpha}$ to denote the ordering obtained from $\alpha$ by inverting it and call $\bar{\alpha}$ the \blue{{\it reversal}} of $\alpha$. \blue{Furthermore, for a subset $S$ of $V$,  the {\it restriction of $\alpha$ to $S$} is the linear ordering, say $\gamma$, on $S$ with the property that $\gamma(v_j)<\gamma(v_{j'})$ if and only if $\alpha(v_j)<\alpha(v_{j'})$ for each pair of elements $v_j,v_{j'}\in S$. We denote $\gamma$ by $\alpha-(V-S)$.} Now, let $\gamma$ be a linear ordering of $V\cup V'$ such that $\gamma(v_j)<\gamma(v_{j'})$ if and only if $\alpha(v_j)<\alpha(v_{j'})$ for each pair of elements $v_j,v_{j'}\in V$. Then $\gamma$ is said to {\it preserve} $\alpha$.
Let $\gamma$ be a linear ordering of $V\cup V'$ that preserves $\alpha$ and $\beta$ and has the property that $\gamma(v_j)<\gamma(v_{j'})$ for each $v_j\in V$ and $v_{j'}\in V'$.
We use $\alpha||\beta$ to denote $\gamma$. Intuitively, $\gamma$ is the concatenation of $\alpha$ followed by $\beta$. 
%Reversely, obtain $\alpha$ from $\gamma$ by deleting each element in $V'$ and preserving the ordering of all remaining elements. We use $\gamma-\beta$ to denote $\alpha$.

\blue{Referring back to Table~\ref{tab:summary}, let  $\Pi_i$ be a subset of $S_3$ for some $i\in\{0,1,\ldots,10\}$}, and let $\cI=(V,\cC)$ be an instance of a ternary permutation constraint satisfaction problem. Furthermore, let $\alpha$ be a linear ordering of $V$. We say that a constraint $(v_1,v_{2},v_{3})$ in $\cC$ is {\it $\Pi_i$-satisfied} by $\alpha$, \blue{if  there is a permutation $\pi\in\Pi$ such that} \steven{$\alpha(v_{\pi(1)})<\alpha(v_{\pi(2)}) < \alpha(v_{\pi(3)})$ where here $\pi$ is assumed to map positions to symbols}.
%Throughout this paper, we consider eleven distinct subsets of $S_3$ that are denoted by %$\Pi_0,\Pi_1,\Pi_2,\ldots,\Pi_{10}$ and summarized in Table~\ref{tab:summary}.
For example,
%$\Pi_5$ contains the two permutations $123$ and $321$. Hence,
if $(v_1,v_2,v_3)$ is $\Pi_5$-satisfied by $\alpha$, then either \steven{$\alpha(v_1)<\alpha(v_2) < \alpha(v_3)$} or \steven{$\alpha(v_3)<\alpha(v_2) < \alpha(v_1)$}.\\

For each $i\in\{0,1,2,\ldots,10\}$, we are now in a position to define the following decision problem.\\

\noindent $k$-$\Pi_i$

\noindent {\bf Instance.} A finite set $V$ of variables and a
% multiset
set $\cC$ of ordered triples of distinct variables from $V$ and a positive integer $k$.

\noindent {\bf Question.} Do there exist at most $k$ linear orderings of $V$ such that each constraint $(v_1,v_2,v_3)$ in $\cC$ is $\Pi_i$-satisfied by one of these orderings.\\

%\noindent 2-$\Pi_i$
%
%\noindent {\bf Instance.} A finite set $V$ of variables and  a multiset $\cC$ of ordered triples of distinct variables from $V$.
%
%\noindent {\bf Question.} Do there exist two linear orderings $\alpha$ and $\beta$ of $V$ such that each constraint $(v_1,v_2,v_3)$ in $\cC$ is $\Pi_i$-satisfied by $\alpha$ or $\beta$.\\

\noindent If the answer to an instance $\cI$ of $k$-$\Pi_i$ is `yes', we say that $\cI$ is {\it $k$-$\Pi_i$-satisfiable} (or {\it $\Pi_i$-satisfiable} for short if $k=1$).

\blue{Lastly, let $\alpha$ be a linear ordering of a set $V$ of variables. We say that $\alpha$ {\it implies a constraint $(v_1,v_{2},v_{3})$ under $\Pi_i$} precisely if $(v_1,v_{2},v_{3})$ is $\Pi_i$-satisfied by $\alpha$.  Furthermore, we use $\cC_{\Pi_i}(\alpha)$ to denote the set of all constraints that are implied by $\alpha$ under $\Pi_i$. }

\subsection{Phylogenetics background}

This section contains preliminaries in the context of phylogenetics. For a more thorough overview, we refer the interested reader to \cite{SempleSteel2003}.

A \emph{binary unrooted phylogenetic tree} of order~$n$ is a tree in which all internal vertices have degree~3 and which has~$n$ leaves that are bijectively labeled with elements in $\n$. For two binary unrooted phylogenetic trees~$T$ and $T'$, we say that~$T$ \emph{displays}~$T'$ if~$T'$ can be obtained from a subtree of~$T$ by suppressing degree-2 vertices. 

A \emph{binary rooted phylogenetic tree} of order~$n$ is a rooted tree in which all edges are directed away from the root which has outdegree 2, all internal vertices have indegree 1 and outdegree~2, and which has~$n$ leaves that are bijectively labeled with elements in~$\n$. If two leaves $a$ and $b$ of a rooted phylogenetic tree $T$ are adjacent to the same parent, then $\{a,b\}$ is called a {\it cherry} of $T$. Furthermore, a binary rooted phylogenetic tree that has exactly one cherry is called a {\it caterpillar}. For two vertices~$u$ and $v$ of a binary rooted phylogenetic tree~$T$, we write~$u<_Tv$ to denote that there exists a directed path from~$u$ to~$v$ in~$T$. Moreover, $\lca_T(u,v)$ denotes the lowest common ancestor of~$u$ and~$v$ in~$T$, i.e.~$\lca_T(u,v)$ is the unique vertex~$w$ such that~$w<_Tu$, $w<_Tv$ and there is no vertex~\steven{$w' \neq w$} such that~$w'<_Tv$, $w'<_Tv$ and~$w<_Tw'$. Again, let $T$ be a rooted phylogenetic tree whose leaves are bijectively labeled with elements in $X$, and let $Y$ be a subset of $X$. We call $X$ the {\it leaf set} of $T$ and denote it by $L(T)$. Furthermore, the {\it minimal rooted subtree} of $T$ that connects all the leaves in $Y$ is denoted by $T(Y)$. Lastly, the {\it restriction of $T$ to $Y$}, denoted by $T|Y$, is the rooted phylogenetic tree obtained from $T(Y)$ by contracting all degree-two vertices apart from the root. 
%Furthermore, we call the rooted phylogenetic tree obtained from the minimal rooted subtree of $T$ that connects all vertices in $Y$  by contracting all non-root degree-$2$ vertices the {\it restriction of $T$ to $Y$} and denoted it by $T|Y$.

Next, we introduce a special type of binary rooted phylogenetic tree. A {\it (rooted) triplet} is a binary rooted phylogenetic tree on three leaves (see Figure \ref{fig:triplet}). 
%For example, for leaves labeled with $a$, $b$, and $c$ we write $ab|c$ or, equivalently, $ba|c$ if the path from $a$ to $b$ does not intersect the path from the root to $c$. 
We say that a binary rooted phylogenetic tree~$T$ \emph{displays} a triplet~$ab|c$ \blue{(or, equivalently, $ba|c$)} if $\lca_T(a,c)=\lca_T(b,c)<_T\lca_T(a,b)$. Moreover, for a triplet $ab|c$, we call $c$ the {\it witness} of $ab|c$. Now, let $\cR$ be a set of triplets. If there exists a rooted phylogenetic tree $T$ such that each triplet in $\cR$ is displayed by $T$, we say that $\cR$ is {\it compatible} and, otherwise, we say that $\cR$ is {\it incompatible.}

\begin{figure}[h]
 \centering
  \includegraphics[width=2cm]{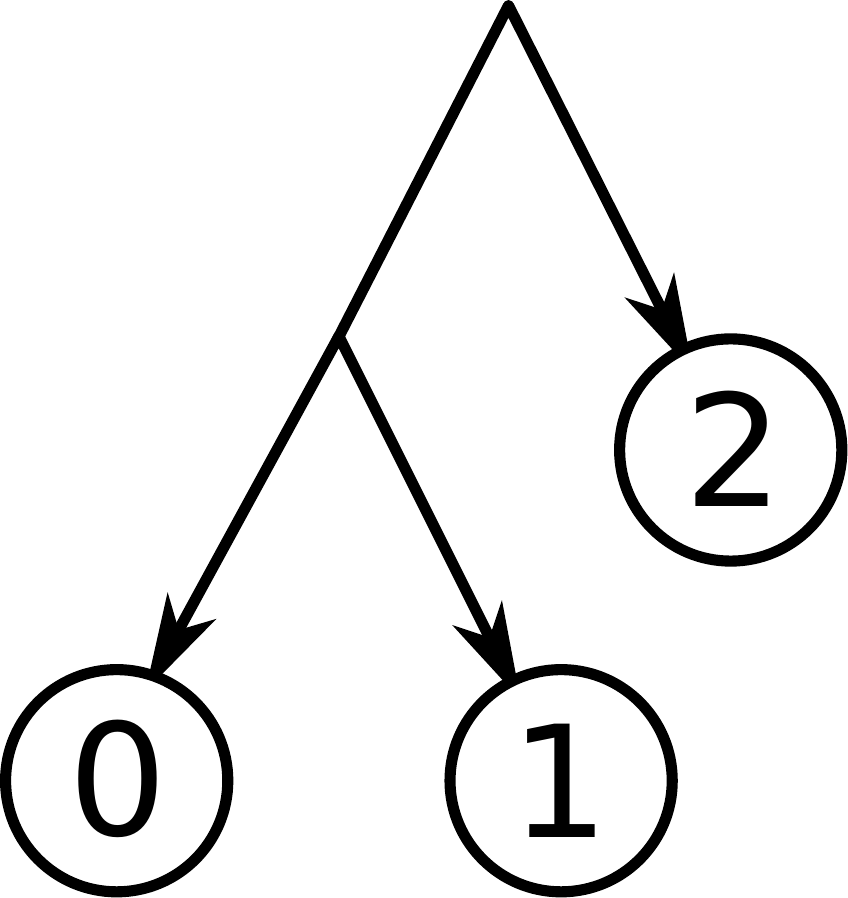}
  \caption{A rooted triplet $01|2$.}
 \label{fig:triplet}
\end{figure}

We are now in a position to state a decision problem that plays an important role in this paper and is strongly related to  \blue{$k$-$\Pi_1$}.\\

\noindent $k$-{\sc Caterpillar Compatibility}

\noindent {\bf Instance.} A set $\cR$ of rooted triplets.

\noindent {\bf Question.} Do there exist at most $k$ caterpillars such that each element in $\cR$ is displayed by at least one such caterpillar?\\

\subsection{A note on $k$-$\Pi_1$ and $k$-{\sc Caterpillar Compatibility}}\label{sec:order2phylo}
A rooted triplet $ab|c$ is a rooted binary phylogenetic tree on three leaves. It is important to note that $a$ and $b$ are indistinguishable from a phylogenetics point of view because a rooted binary phylogenetic tree $T$ on three leaves $a$, $b$, and $c$ and with $\{a,b\}$ being the cherry of $T$ such that $a$ is the right and $b$ is the left child of the common parent of $a$ and $b$ is considered to be the same as the tree obtained from $T$ by swapping $a$ and $b$. 

Let $(1,2,\ldots,n-1,n)$ denote a caterpillar $T$ on $n$ leaves whose cherry is $\{n-1,n\}$ and the path from each \blue{leaf labeled $i\in\{1,2,\ldots,n-1\}$} to the root of $T$ has length $i$ \blue{while the path from the leaf labeled $n$  to the root of $T$ to has length $n-1$}. Since $n$ and $n-1$ are indistinguishable, $T$ naturally corresponds to the two linear ordering \blue{$(n,n-1,n-2,\ldots,2,1)$ and $(n-1,n,n-2,\ldots,2,1)$}. Moreover, each constraint $(a,b,c)$ in an instance of \blue{$k$-$\Pi_1$} can be satisfied by an ordering $\alpha$ with $\alpha(a)<\alpha(b)<\alpha(c)$ or $\alpha(a)<\alpha(c)<\alpha(b)$ and corresponds to a rooted triplet $bc|a$ that can be displayed by a caterpillar whose distance from $a$ to the root is shorter than the distance from $b$ to the root and also shorter than the distance from $c$ to the root. We summarize the strong relationship between the two problems in the following observation.

\begin{observation}\label{obs}
The problem $k$-$\Pi_1$  is NP-complete if and only if $k$-{\sc Caterpillar Compatibility} is NP-complete.
\end{observation}

\section{Hardness results for all eleven CSP problems on two linear orderings}\label{sec:two}
In this section, we settle the complexity of each ternary permutation CSP problem 2-$\Pi_i$ with $i\in\{0,1,2,\ldots,10\}$. We start with the following observation.
%an observation that shows that 2-$\Pi_i$ is polynomial-time solvable for each $i \in %\{2,3,7,8,10\}$.

\steven{
\begin{observation}
For each $i \in \{2,3,7,8,10\}$, the problem \em 2-$\Pi_i$ is trivially polynomial-time solveable. In particular, every instance is a `yes'-instance.
\end{observation}
\begin{proof}
It can easily be verified that, for each of the described problems, $\{ \alpha, \bar{\alpha} \}$ is a valid solution, for \emph{any} linear ordering $\alpha$.
\end{proof}
}

The next \blue{observation is easily verified and implicitly used throughout the remainder of this section.}
% shows that, for each $i \in \{0,1,2,\ldots,10\}$, the decision problem 2-$\Pi_i$ is in NP. The result is subsequently used to establish NP-completeness of 2-$\Pi_i$ for each $i\in\{0,1,4,5,6,9\}$.

\begin{observation}\label{obs:np}
For each $i\in\{0,1,2,\ldots, 10\}$, the problem {\em 2}-$\Pi_i$ is in NP.
\end{observation}

%\begin{proof}
%Let $\cI=(V,\cC)$ be an instance of 2-$\Pi_i$, and let $\alpha$ and $\beta$ be two linear orderings of $V$. It can be verified in polynomial time whether or not each constraint $(v_1,v_2,v_3)$ in $\cC$ is $\Pi_i$-satisfied by $\alpha$ or $\beta$.
%\end{proof}

%%%%%%%%%%%%%%%%%%%%%%%%%%%%%%%%%%%%%%%%%%%%%%%%%%%%%%%%%%%%%%%%%%%%%%%%%%%%%%%%%%%

\begin{theorem}\label{t:2-pi0}
The problem {\em 2}-$\Pi_0$ is NP-complete.
\end{theorem}

\begin{proof}
To establish the result, we use a polynomial-time reduction from 1-$\Pi_5$. Let $\cI=(V,\cC)$ be an instance of 1-$\Pi_5$. Let $\cC'$ be the set of ternary constraints in which each constraint in $\cC$ is represented by two constraints. In particular, set $$\cC'=\bigcup_{(v_1,v_2,v_3)\in\cC}\{(v_1,v_2,v_3),(v_3,v_2,v_1)\}.$$ Now, let $\cI'=(V,\cC')$ be an instance of 2-$\Pi_0$. As $|\cC'|=2|\cC|$, we have that $\cI'$ has polynomial size and that the reduction can be carried out in polynomial time. \blue{We now claim that $\cI$ is 1-$\Pi_5$-satisfiable if and only if $\cI'$ is 2-$\Pi_0$-satisfiable.} 

Suppose that $\cI$ is 1-$\Pi_5$-satisfiable. Let $\alpha$ be a linear ordering of $V$ that satisfies $\cI$. It is now easily checked that, for each constraint $(v_1,v_2,v_3)$ in $\cC$, either $\alpha(v_1)<\alpha(v_2)<\alpha(v_3)$ in which case $\bar{\alpha}(v_3)<\bar{\alpha}(v_2)<\bar{\alpha}(v_1)$ or $\alpha(v_3)<\alpha(v_2)<\alpha(v_1)$ in which case $\bar{\alpha}(v_1)<\bar{\alpha}(v_2)<\bar{\alpha}(v_3)$. Hence, $\alpha$ and $\bar{\alpha}$ are a solution to $\cI'$ and, so, \blue{$\cI'$ is 2-$\Pi_0$-satisfiable}.

On the other hand, suppose that $\cI'$ is 2-$\Pi_0$-satisfiable. Let $\alpha$ and $\beta$ be two linear orderings of $V$ that satisfy $\cI'$, and let $(v_1,v_2,v_3)$ be an element of $\cC$. Then exactly one element in $\{(v_1,v_2,v_3),(v_3,v_2,v_1)\}$ is $\Pi_0$-satisfied by $\alpha$. Hence, $\alpha(v_1)<\alpha(v_2)<\alpha(v_3)$ or $\alpha(v_3)<\alpha(v_2)<\alpha(v_1)$. Both cases imply that $(v_1,v_2,v_3)$ is $\Pi_5$-satisfied by $\alpha$. Hence, $\alpha$ is a solution to $\cI$ and, so $\cI$ is \blue{1-$\Pi_5$-satisfiable}. \blue{The theorem now follows.}
\end{proof}

%%%%%%%%%%%%%%%%%%%%%%%%%%%%%%%%%%%%%%%%%%%%%%%%%%%%%%%%%%%%%%%%%%%%%%%%%%%%%%%%%%%
\begin{theorem}\label{t:2-pi1}
The problem {\em 2-}$\Pi_1$ is NP-complete.
\end{theorem}

\begin{proof}
Let $\cI=(V,\cC)$ be an instance of 2-$\Pi_0$, where $\cC=\{C^1,C^2,\ldots,C^n\}$. Furthermore, for each $i\in\{1,2,\ldots,n\}$, let $C^i=(v_1^i,v_2^i,v_3^i)$, with $\{v_1^i,v_2^i,v_3^i\}\subseteq V$, and let $v_d^i$ and $v_e^i$ be two new variables,
not contained in $V$, corresponding to \blue{each constraint}.
% that \steven{is} not contained in $V$.
To show that the theorem holds, we reduce $\cI$ to an instance of 2-$\Pi_1$. Set $$V'=V\cup\{v_d^1,v_e^1,v_d^2,v_e^2,\ldots,v_d^n,v_e^n\},$$ and set $$\cC'=\bigcup_{C^i\in\cC}\{(v_1^i,v_2^i,v_d^i),(v_2^i,v_3^i,v_e^i),(v_e^i,v_1^i,v_2^i),(v_d^i,v_1^i,v_2^i),(v_1^i,v_e^i,v_d^i)\}.$$ Now, let $\cI'=(V',\cC')$ be an instance of 2-$\Pi_1$. Since $|V'|=|V|+2n$ and $|\cC'|=5|\cC|$, the reduction can clearly be carried out in polynomial time and has polynomial size. The remainder of the proof consists of establishing that $\cI$ is 2-$\Pi_0$-satisfiable if and only if $\cI'$ is 2-$\Pi_1$-satisfiable. 

First, suppose that $\cI$ is 2-$\Pi_0$-satisfiable. Let $\alpha$ and $\beta$ be two linear orderings of $V$ that satisfy $\cI$. Let $$W=\{v_d^i,v_e^i:C^i \textnormal{ is }\Pi_0\textnormal{-satisfied by }\alpha\}$$ and, similarly, let $$W'=\{v_d^i,v_e^i:C^i \textnormal{ is not }\Pi_0\textnormal{-satisfied by }\alpha\}.$$ Furthermore, let $\gamma$ be \blue{an arbitrary} linear ordering of $W$, and let $\gamma'$ be \blue{an arbitrary} linear ordering of $W'$. Then $\alpha'=\gamma'||\alpha||\gamma$ and $\beta'=\gamma||\beta||\gamma'$ are two linear orderings on $V'$. Now, for each $C^i=(v_1^i,v_2^i,v_3^i)$ that is $\Pi_0$-satisfied by $\alpha$ (resp. $\beta$), the three  constraints $(v_1^i,v_2^i,v_d^i)$, $(v_2^i,v_3^i,v_e^i)$, and  $(v_1^i,v_e^i,v_d^i)$ are $\Pi_1$-satisfied by $\alpha'$ (resp $\beta'$) while the two constraints $(v_e^i,v_1^i,v_2^i)$ and $(v_d^i,v_1^i,v_2^i)$ are $\Pi_1$-satisfied by $\beta'$ (resp. $\alpha'$). Hence $\cI'$ is 2-$\Pi_1$-satisfiable. 

Second, suppose that $\cI'$ is 2-$\Pi_1$-satisfiable. Let $\alpha'$ and $\beta'$ be two linear orderings of $V'$ that satisfy $\cI'$. Assume that, for some $i\in\{1,2,\ldots,n\}$, the two constraints $(v_1^i,v_2^i,v_d^i)$ and $(v_2^i,v_3^i,v_e^i)$ are not both $\Pi_1$-satisfied by exactly one of $\alpha'$ and $\beta'$. Then without loss of generality, we may assume that $(v_1^i,v_2^i,v_d^i)$ is $\Pi_1$-satisfied by $\alpha'$ and that $(v_2^i,v_3^i,v_e^i)$ is $\Pi_1$-satisfied by $\beta'$. Since $\beta'(v_2^i)<\beta'(v_e^i)$, it follows that  $(v_e^i,v_1,v_2^i)$ is $\Pi_1$-satisfied by $\alpha'$. Similarly, since $\alpha'(v_1^i)<\alpha'(v_d^i)$,  it follows that  $(v_d^i,v_1^i,v_2^i)$ is $\Pi_1$-satisfied by $\beta'$. Moreover, since $\alpha'(v_e^i)<\alpha'(v_1^i)$ and $\beta'(v_d^i)<\beta'(v_1^i)$, this implies that  neither $\alpha'$ nor by $\beta$  $\Pi_1$-satisfies  $(v_1^i,v_e^i,v_d^i)$. Thus, $(v_1^i,v_2^i,v_d^i)$ and $(v_2^i,v_3^i,v_e^i)$ are  both $\Pi_1$-satisfied by either $\alpha'$ or $\beta'$. 
Hence, we have $\alpha'(v_1^i)<\alpha'(v_2^i)<\alpha'(v_3^i)$ or $\beta'(v_1^i)<\beta'(v_2^i)<\beta'(v_3^i)$. It now follows that $\cI$ is 2-$\Pi_0$-satisfied by the two linear orderings \blue{$\alpha'$ restricted to $V$ and $\beta'$ restricted to $V$. The theorem now follows.}
\end{proof}

%%%%%%%%%%%%%%%%%%%%%%%%%%%%%%%%%%%%%%%%%%%%%%%%%%%%%%%%%%%%%%%%%%%%%%%%%%%%%%%%%%%
\begin{theorem}\label{t:2-pi_4}
The problem {\em 2-}$\Pi_4$ is NP-complete.
\end{theorem}

\begin{proof}
To establish the result, we use a polynomial-time reduction from 1-$\Pi_9$. Let $\cI=(V,\cC)$ be an instance of 1-$\Pi_9$. Let $\cC'$ be the set of ternary constraints in which each constraint in $\cC$ is represented by two constraints. In particular, set $$\cC'=\bigcup_{(v_1,v_2,v_3)\in\cC}\{(v_2,v_1,v_3),(v_2,v_3,v_1)\}.$$ Now, let $\cI'=(V,\cC')$ be an instance of 2-$\Pi_4$. As $|\cC'|=2|\cC|$, we have that $\cI'$ has polynomial size and that the reduction can be carried out in polynomial time. We now claim that $\cI$ is 1-$\Pi_9$-satisfiable if and only if $\cI'$ is 2-$\Pi_4$-satisfiable.

Suppose that $\cI$ is 1-$\Pi_9$-satisfiable. Let $\alpha$ be a linear ordering of $V$ that satisfies $\cI$. It follows that, for each constraint $(v_1,v_2,v_3)$ in $\cC$, either $\alpha(v_2)<\alpha(v_i)<\alpha(v_j)$, or $\alpha(v_{i'})<\alpha(v_{j'})<\alpha(v_2)$ with $\{i,j\}=\{i',j'\}=\{1,3\}$. Moreover, in both cases, it is easily checked that exactly one of $(v_2,v_1,v_3)$ and $(v_2,v_3,v_1)$ is $\Pi_4$-satisfied by $\alpha$ while the other constraint is  $\Pi_4$-satisfied by $\bar{\alpha}$. Hence, $\alpha$ and $\bar{\alpha}$ are a solution to $\cI'$ and, so, $\cI'$ is 2-$\Pi_4$-satisfiable.

Now, suppose that $\cI'$ is 2-$\Pi_4$-satisfiable. Let $\alpha$ and $\beta$ be two linear orderings of $V$ that satisfy $\cI'$, and let $(v_1,v_2,v_3)$ be an element of $\cC$. Then exactly one element in $\{(v_2,v_1,v_3),(v_2,v_3,v_1)\}$ is $\Pi_4$-satisfied by $\alpha$. In particular, this implies that exactly one of the following holds:
\begin{itemize}
\item [(i)] $\alpha(v_2)<\alpha(v_1)<\alpha(v_3)$,
\item [(ii)] $\alpha(v_1)<\alpha(v_3)<\alpha(v_2)$,
\item [(iii)] $\alpha(v_2)<\alpha(v_3)<\alpha(v_1)$, or
\item [(iv)] $\alpha(v_3)<\alpha(v_1)<\alpha(v_2)$.
\end{itemize}
Regardless of which of (i)-(iv) holds, $(v_1,v_2,v_3)$ is $\Pi_9$-satisfied by $\alpha$. Hence, $\alpha$ is a solution to $\cI$. The theorem now follows.
\end{proof}

%%%%%%%%%%%%%%%%%%%%%%%%%%%%%%%%%%%%%%%%%%%%%%%%%%%%%%%%%%%%%%%%%%%%%%%%%%%%%%%%%%%

\begin{lemma}\label{l:2-Pi_5-caterpillar}
\blue{Let $V=\{1,2,3,4,5\}$, and let $\gamma=(1,2,3,4,5)$ and $\gamma'=(5,2,3,4,1)$ be two linear orderings of $V$. Then the instance $\cI=(V,\cC_{\Pi_5}(\gamma)\cup\cC_{\Pi_5}(\gamma'))$ of {\em 2-}$\Pi_5$ has a unique solution (up to reversal). In particular, each solution of $\cI$ consists of \steven{one element from $\{ \gamma, \bar \gamma\}$ and
one element from $\{ \gamma', \bar \gamma' \}$.}}
\end{lemma}

\begin{proof}
\steven{Computational proof (see appendix).}
\end{proof}

\begin{theorem}
The problem {\em 2-}$\Pi_5$ is NP-complete.
\end{theorem}

\begin{proof}
\blue{Throughout the proof, let $\gamma=(1,2,3,4,5)$ and $\gamma'=(5,2,3,4,1)$ be two linear orderings of $\{1,2,3,4,5\}$.}
Let $\cI=(V,\cC)$ be an instance of 1-$\Pi_5$, where $\cC=\{C^1,C^2,\ldots,C^n\}$. Furthermore, for each $i\in\{1,2,\ldots,n\}$, let $C^i=(v_1^i,v_2^i,v_3^i)$, with $\{v_1^i,v_2^i,v_3^i\}\subseteq V$, and let $v_d^i$ and $v_e^i$ be two new variables,
not contained in $V$, \blue{for each constraint.} To show that the theorem holds, we reduce $\cI$ to an instance of 2-$\Pi_5$. Let $D=\{v_d^1,v_e^1,v_d^2,v_e^2,\ldots,v_d^n,v_e^n\}$, and let $$V'=V\cup D\cup\{1,2,3,4,5\},$$ \blue{where $\{1,2,3,4,5\}$ neither intersects with $D$ nor $V$.} Furthermore, we define the following four new sets of constraints.
\begin{itemize}
\item [(i)] Let \blue{$\cC_1=\cC_{\Pi_5}(\gamma)\cup\cC_{\Pi_5}(\gamma')$.}
\item [(ii)] Let $\cC_2=\bigcup_{C^i\in\cC}\{(v_1^i,v_d^i,v_3^i),(v_1^i,v_e^i,v_3^i),(v_d^i,v_2^i,v_e^i)\}$.
\item [(iii)] Let $\cC_3=\bigcup_{v_j\in V}\{(3,v_j,4),(4,v_j,5),(1,v_j,2),(1,v_j,3)\}$.
\item [(iv)] Let $\cC_4=\bigcup_{i\in \{1,2,\ldots,n\}}\{(2,v_d^i,3),(2,v_e^i,3),(1,v_d^i,2),(1,v_e^i,2),(4,v_d^i,5),(4,v_e^i,5),(3,v_d^i,5),(3,v_e^i,5)\}$.
\end{itemize}
Now, let $\cC'=\cC_1\cup\cC_2\cup\cC_3\cup\cC_4$, and let $\cI'=(V',\cC')$ be an instance of 2-$\Pi_5$. Since $\cC_1$ contains a constant number of constraints it is easily checked that $\cI'$ has size polynomial in $|V|$ and $n$ and, so, the reduction can be carried out in polynomial time. We now claim that $\cI$ is 1-$\Pi_5$-satisfiable if and only if $\cI'$ is 2-$\Pi_5$-satisfiable.

First, suppose that $\cI$ is 1-$\Pi_5$-satisfiable. Let $\alpha$ be a linear ordering of $V$ that satisfies $\cI$. Let $\delta$ be a linear ordering of $V\cup D$ such that $$\delta(v_1^i)<\delta(v_d^i)<\delta(v_2^i)<\delta(v_e^i)<\delta(v_3^i) \textnormal{ or }\delta(v_3^i)<\delta(v_e^i)<\delta(v_2^i)<\delta(v_d^i)<\delta(v_1^i)$$ for each $i\in\{1,2,\dots,n\}$. Since $\alpha$ is a solution to $\cI$, note that $\delta$ exists.  

Now, let $$\alpha'=(1,2,3,4,5)||\delta,$$ and let  \blue{$$\beta'=(5,2)||\delta-V||(3)||\alpha||(4,1)$$} be two linear orderings of $V'$. We next argue that each constraint in $\cC'$ is $\Pi_5$-satisfied by $\alpha'$ or $\beta'$. Since $\alpha'$ preserves $\gamma$ and since $\beta'$ preserves $\gamma'$, it follows that each constraint in $\cC_1$ is $\Pi_5$-satisfied by $\alpha'$ or $\beta'$. Furthermore, for each $C^i\in \cC$, the three corresponding constraints in $\cC_2$ are, by construction, $\Pi_5$-satisfied by $\alpha'$. Turning to the constraints in $\cC_3$, we observe that \blue{
$\max(\beta'(2),\beta'(3),\beta'(5))<\beta'(v_j)$ and $\beta'(v_j)<\min(\beta'(1),\beta'(4))$ and, hence,} all four constraints that correspond to $v_j$ in $\cC_3$ are $\Pi_5$-satisfied by $\beta'$. \blue{Similarly, for $k\in\{d,e\}$, a straightforward check shows that} all eight constraints that correspond to $v_k^i$ in $\cC_4$ are $\Pi_5$-satisfied by $\beta'$. Now, as each constraint in $\cC'$ is $\Pi_5$-satisfied by $\alpha'$ or $\beta'$, we deduce that $\cI'$ is 2-$\Pi_5$-satisfiable.

Second, suppose that $\cI'$ is 2-$\Pi_5$-satisfiable. Let $\alpha'$ and $\beta'$ be two linear orderings of $V'$ that satisfy $\cI'$. Note that, by Lemma~\ref{l:2-Pi_5-caterpillar}, each solution to the instance $(\{1,2,3,4,5\},\cC_1)$ of 2-$\Pi_5$ consists of
\steven{one element from $\{ \gamma, \bar \gamma\}$ and
one element from $\{ \gamma', \bar \gamma' \}$.}
Assume \steven{for the time being} that
$\alpha'$ preserves $\gamma$ and that $\beta'$ preserves $\gamma'$. Now assume that $(3,v_j,4)$ is $\Pi_5$-satisfied by $\alpha'$ for some $v_j\in V$. Then each constraint in $\{(4,v_j,5),(1,v_j,2),(1,v_j,3)\}$ \blue{and, hence, $(3,v_j,4)$} is $\Pi_5$-satisfied by $\beta'$. \blue{On the other hand, assume that $(3,v_j,4)$ is not $\Pi_5$-satisfied by $\alpha'$ for some $v_j\in V$. Then, $(3,v_j,4)$ is $\Pi_5$-satisfied by $\beta'$. Thus, regardless of whether $(3,v_j,4)$ is $\Pi_5$-satisfied by $\alpha'$ or not, we have $\beta'(3)<\beta'(v_j)<\beta'(4)$}. Similarly, assume that $(2,v_k^i,3)$ is $\Pi_5$-satisfied by $\alpha'$ for some $k\in\{d,e\}$ and $i \in \{1,2,\ldots,n\}$. Then each constraint in $\{(1,v_k^i,2),(4,v_k^i,5),(3,v_k^i,5)\}$ \blue{and, hence, $(2,v_k^i,3)$} is $\Pi_5$-satisfied by $\beta'$. \blue{On the other hand, assume that $(2,v_k^i,3)$ is not $\Pi_5$-satisfied by $\alpha'$ for some $v_k^i\in V$. Then, $(2,v_k^i,3)$ is $\Pi_5$-satisfied by $\beta'$. Thus, regardless of whether $(2,v_k^i,
3)$ is $\Pi_5$-satisfied by $\alpha'$ or not, we have $\beta'(2)<\beta'(v_k^i)<\beta'(3)$}. 
In summary, it follows that each constraint in \blue{$\cC_3\cup\cC_4$} is $\Pi_5$-satisfied by $\beta'$. \blue{Now, since $\beta'(v_k^i)<\beta'(3)<\beta'(v_j)$ for each $i\in\{1,2,\ldots,n\}$, $k\in\{d,e\}$, and $v_j\in V$, it follows that,}
for each $C^i\in \cC$, the three constraints in $\{(v_1^i,v_d^i,v_3^i),(v_1^i,v_e^i,v_3^i),(v_d^i,v_2^i,v_e^i)\}$ are $\Pi_5$-satisfied by $\alpha'$. It is now straightforward to check that $\alpha'(v_1^i)<\alpha'(v_2^i)<\alpha'(v_3^i)$ or $\alpha'(v_3^i)<\alpha'(v_2^i)<\alpha'(v_1^i)$. Hence, $\alpha'-(\{1,2,3,4,5\}\cup D)$ is a linear ordering of $V$ that 1-$\Pi_5$-satisfies each constraint in $\cC$ and, therefore, we have that $\cI$ is 1-$\Pi_5$-satisfiable. \steven{We complete the proof of the converse by noting that 
symmetrical arguments can be used to show that $\cI$ is 1-$\Pi_5$-satisfiable if
$\alpha'$ preserves $\bar \gamma$ (rather than $\gamma$) and/or $\beta'$ preserves
$\bar \gamma'$ (rather than $\gamma'$.)}
The theorem now follows by combining both cases.
\end{proof}

\begin{lemma}\label{l:2-Pi_6-caterpillar}
Let $V=\{1,2,3,4\}$, and let $\gamma=(1,2,3,4)$ and $\gamma'=(2,4,1,3)$ be two linear orderings of $V$. Then the instance $\cI=(V,\cC_{\Pi_6}(\gamma)\cup\cC_{\Pi_6}(\gamma'))$ of {\em 2-}$\Pi_6$ has a unique solution. In particular, $\gamma$ and $\gamma'$ are a solution of $\cI$.
\end{lemma}

\begin{proof}
\steven{Computational proof (see appendix).}
\end{proof}

\begin{theorem}
The problem {\em 2-}$\Pi_6$ is NP-complete.
\end{theorem}

\begin{proof}
Throughout the proof, let $\gamma=(1,2,3,4)$ and $\gamma'=(2,4,1,3)$ be two linear orderings of $\{1,2,3,4\}$.
Let $\cI=(V,\cC)$ be an instance of 2-$\Pi_1$ with $V=\{v_1,v_2,\ldots,v_m\}$ and $\cC=\{C^1,C^2,\ldots,C^n\}$. Furthermore, for each $i\in\{1,2,\ldots,n\}$, let $C^i=(v_1^i,v_2^i,v_3^i)$, with $\{v_1^i,v_2^i,v_3^i\}\subseteq V$.  Lastly, let $D=\{1_{v_1},1_{v_2},\ldots,1_{v_m}\}$ such that $D\cap\{1,2,3,4\}=\emptyset$. To show that the result holds, we reduce $\cI$ to an instance $\cI'$ of 2-$\Pi_6$ in the following way. For each $v_j\in V$, we use $\cC(v_j)$ to denote the set obtained from $\cC_{\Pi_6}(\gamma)\cup\cC_{\Pi_6}(\gamma')$ by replacing each occurrence of the variable $1$ with $1_{v_j}$. Let $V'=D \cup\{2,3,4\}$ be a set of variables. We next define two new sets of constraints. In particular, let $$\cC_1=\bigcup_{v_j\in V}\cC(v_j),$$ and let $$\cC_2=\bigcup_{C^i\in\cC}\{(1_{v_1^i},1_{v_2^i},1_{v_3^i}),(1_{v_1^i},1_{v_3^i},1_{v_2^i}),(1_{v_1^i},1_{v_2^i},3)\},$$ where each of $1_{v_1^i}$, $1_{v_2^i}$, and \steven{$1_{v_3^i}$} is an element in $D$.

Now, let $\cI'=(V',\cC_1\cup\cC_2)$ and observe that each constraint in $\cC_1\cup\cC_2$ consists indeed of three elements in $V'$. Moreover, since $\cC_1$ contains a number of constraints that is polynomial in $m$, it is easily checked that $\cI'$ has size polynomial in $m$ and $n$ and, so, the reduction can be carried out in polynomial time.  To complete the proof, we show that $\cI$ is 2-$\Pi_1$-satisfiable if and only if $\cI'$ is 2-$\Pi_6$-satisfiable. 

First, suppose that $\cI$ is 2-$\Pi_1$-satisfiable. Then there exist two linear orderings $\alpha$ and $\beta$ on $V$ such that each $C^j\in \cC$ is $\Pi_1$-satisfied by $\alpha$ or $\beta$. Let $\alpha_{1}$ be the linear ordering of $D$ obtained from $\alpha$ by replacing each $v_j$ with $1_{v_j}$ and, similarly, let $\beta_{1}$ be the linear ordering of $D$ obtained from $\beta$ by replacing each $v_j$ with $1_{v_j}$. Now, let $$\alpha'=\alpha_{1}||(2,3,4),$$ and let $$\beta'=(2,4)||\beta_{1}||(3)$$ be two linear orderings on $V'$. We next show that each constraint in $\cC_1\cup\cC_2$ is $\Pi_6$-satisfied by $\alpha'$ or $\beta'$. Since no constraint in $\cC_1$ contains two elements of $D$, it follows by Lemma~\ref{l:2-Pi_6-caterpillar} and regarding $1$ as a placeholder for $\alpha_1$ (resp. $\beta_1$), that each constraint in $\cC_1$ is $\Pi_6$-satisfied by $\alpha'$ or $\beta'$. Turning to the constraints in $\cC_2$, we have that, if $C^i\in\cC$ is $\Pi_1$-satisfied by $\alpha$, then the first two 
constraints that correspond to $C^i$ in $\cC_2$ are $\Pi_6$-satisfied by $\alpha'$ while, if $C^i$ is $\Pi_1$-satisfied by $\beta$, then the first two constraints that correspond to $C^i$ in $\cC_2$ are $\Pi_6$-satisfied by $\beta'$. Moreover, since $\max(\alpha'(1_{v_1^i}),\alpha'(1_{v_2^i}))<\alpha'(3)$ and $\max(\beta'(1_{v_1^i}),\beta'(1_{v_2^i}))<\beta'(3)$, it follows that $(1_{v_1^i},1_{v_2^i},3)$ is $\Pi_6$-satisfied by $\alpha'$ or $\beta'$. Now, as each constraint in $\cC_1\cup\cC_2$ is $\Pi_6$-satisfied by $\alpha'$ or $\beta'$, we deduce that $\cI'$ is 2-$\Pi_6$-satisfiable.

Second, suppose that $\cI'$ is 2-$\Pi_6$-satisfiable. Then, there exist two linear orderings $\alpha'$ and $\beta'$ on $V'$ such that each constraint in $\cC_1\cup\cC_2$ is $\Pi_6$-satisfied by $\alpha'$ or $\beta'$. We next show that, for each $1_{v_j}\in D$ with $j\in\{1,2,\ldots,m\}$, we have $\alpha'(1_{v_j})<\alpha'(2)$ and $\beta'(4)<\beta'(1_{v_j})<\beta'(3)$ (up to interchanging the roles of $\alpha'$ and $\beta'$). For a contradiction, assume that this is not the case. Then, there exists an element $1_{v_{j'}}\in D$, such that one of the followings holds:
\begin{itemize}
\item [(i)] $\alpha'(2)<\alpha'(1_{v_{j'}})$, 
\item [(ii)] $\beta'(1_{v_{j'}})<\min(\beta'(3),\beta'(4))$,
\item [(iii)] $\max(\beta'(3),\beta'(4))<\beta'(1_{v_{j'}})$, or
\item [(iv)] $\beta'(3)<\beta'(1_{v_{j'}})<\beta'(4)$. 
\end{itemize}
Let $\delta$ be the restriction of $\alpha'$ to $\{1_{v_{j'}},2,3,4\}$ and, similarly, let $\delta'$ be the restriction of $\beta'$ to the same four-element set. Regardless of which of (i), (ii), (iii), or (iv) holds, we obtain two linear orderings on \blue{$\{1_{v_{j'}},2,3,4\}$} of which at least one is different from $(1_{v_{j'}},2,3,4)$ and $(2,4,1_{v_{j'}},3)$. This contradicts the fact that, by Lemma~\ref{l:2-Pi_6-caterpillar}, the instance $(\{1_{v_{j'}},2,3,4\},\cC(v_{j'}))$ of 2-$\Pi_6$, whose set of constraints is a subset of $\cC'$, has the unique solution $(1_{v_{j'}},2,3,4)$ and $(2,4,1_{v_{j'}},3)$. For the remainder of the proof, we may therefore assume that, for each $1_{v_j}\in D$ with $j\in\{1,2,\ldots,m\}$, we have $\alpha'(1_{v_j})<\alpha'(2)$ and $\beta'(4)<\beta'(1_{v_j})<\beta'(3)$. 

Now, let $\alpha$ be the linear ordering of $V$ that is obtained from $\alpha'-\{2,3,4\}$ by replacing each $1_{v_j}$ with $v_j$ for $j\in\{1,2,\ldots,m\}$. Similarly, let $\beta$ be the linear ordering of $V$ that is obtained from $\beta'-\{2,3,4\}$ by replacing each $1_{v_j}$ with $v_j$. Consider an element $C^i=(v_1^i,v_2^i,v_3^i)$ in $\cC$ with $i\in\{1,2,\ldots,n\}$ and its three corresponding constraints in $\cC_2$. If $\alpha'(1_{v_1^i})<\min(\alpha'(1_{v_2^i}),\alpha'(1_{v_3^i}))$ or $\beta'(1_{v_1^i})<\min(\beta'(1_{v_2^i}),\beta'(1_{v_3^i}))$, then it is easily checked that $C^i$ is $\Pi_1$-satisfied by $\alpha$ or $\beta$. Otherwise, as $\alpha'$ and $\beta'$ is a solution to $\cI'$, we have $\alpha'(1_{v_2^i})<\alpha'(1_{v_3^i})<\alpha(1_{v_1^i})$ and $\beta'(1_{v_3^i})<\beta'(1_{v_2^i})<\beta'(1_{v_1^i})$ (up to interchanging the roles of $\alpha'$ and $\beta'$). In particular, by Lemma~\ref{l:2-Pi_6-caterpillar} and the argument in the last paragraph, we have $$\alpha'(1_{v_2^i})<\alpha'(1_{v_
3^i})<\alpha(1_{v_1^i})<\alpha'(2)<\alpha'(3)<\alpha'(4)$$ and $$\beta'(2)<\beta'(4)<\beta'(1_{v_3^i})<\beta'(1_{v_2^i})<\beta'(1_{v_1^i})<\beta'(3).$$ It now follows that $(1_{v_1^i},1_{v_2^i},3)$ is neither $\Pi_6$-satisfied by $\alpha'$ nor  $\Pi_6$-satisfied by $\beta'$; a contradiction. Thus,  we have $\alpha'(1_{v_1^i})<\min(\alpha'(1_{v_2^i}),\alpha'(1_{v_3^i}))$ or $\beta'(1_{v_1^i})<\min(\beta'(1_{v_2^i}),\beta'(1_{v_3^i}))$; thereby implying that $C^i$ is $\Pi_1$-satisfied by $\alpha$ or $\beta$. Thus, $\cI$ is 2-$\Pi_1$-satisfiable.

The theorem now follows by combining both cases.
\end{proof}

\begin{lemma}\label{l:2-Pi_9-caterpillar}
Let $V=\{1,2,3,4,5,6,7\}$, and let $\gamma=(1,2,3,4,5,6,7)$ and $\gamma'=(2,5,7,3,1,6,4)$ be two linear orderings of $V$. Then the instance $\cI=(V,\cC_{\Pi_9}(\gamma)\cup\cC_{\Pi_9}(\gamma'))$ of {\em 2-}$\Pi_9$ has a unique solution (up to reversal). 
\steven{In particular, each solution of $\cI$ consists of one element from $\{ \gamma, \bar \gamma\}$ and one element from $\{ \gamma', \bar \gamma' \}$.}
%\blue{In particular, each solution of $\cI$ consists of two elements in $\{\gamma,%\gamma'\bar\gamma,\bar\gamma'\}$}. 
%\steven{Same point as in the earlier betweenness proof holds here.}
\end{lemma}
\begin{proof}
\steven{Computational proof (see appendix).}
\end{proof}

\begin{theorem}
The problem {\em 2-}$\Pi_9$ is NP-complete.
\end{theorem}

\begin{proof}
Let $\cI=(V,\cC)$ be an instance of 1-$\Pi_5$ with $\cC=\{C^1,C^2,\ldots,C^n\}$. Furthermore, for each $i\in\{1,2,\ldots,n\}$, let $C^i=(v_1^i,v_2^i,v_3^i)$, with $\{v_1^i,v_2^i,v_3^i\}\subseteq V$. We next reduce $\cI$ to an instance of 2-$\Pi_9$. For each $C^i\in\cC$, let $V^i$ be the set $\{1^i,2^i,3^i,5^i,v_1^i,v_2^i,v_3^i\}$ of variables, let $\gamma^i=(1^i,2^i,3^i,v_1^i,5^i,v_2^i,v_3^i)$ and $\delta^i=(2^i,5^i,v_3^i,3^i,1^i,v_2^i,v_1^i)$ be two linear orderings of $V^i$. By replacing the variables $4$, $6$, and $7$ with $v_1^i$, $v_2^i$, and $v_3^i$, respectively, in the statement of Lemma~\ref{l:2-Pi_9-caterpillar}, note that \blue{a solution to} the instance $(V^i,\cC_{\Pi_9}(\gamma^i)\cup\cC_{\Pi_9}(\delta^i))$ of 2-$\Pi_9$
\steven{consists of one element from $\{\gamma^i, \bar \gamma^i\}$ and
one element from $\{\delta^i, \bar \delta^i\}$}.
%\blue{consists of two elements in $\{\gamma^i,\delta^i,\bar\gamma^i,\bar\delta^i\}$.}
Now, let $\cI'=(V',\cC')$ be an instance of 2-$\Pi_9$ with $$V'=\bigcup_{i\in\{1,2,\ldots,n\}}V^i$$ and $$\cC'=\bigcup_{i\in\{1,2,\ldots,n\}}(\cC_{\Pi_9}(\gamma^i)\cup\cC_{\Pi_9}(\delta^i)).$$ Since the number of elements in  $\cC'$ and $V'$ is polynomial in $n$, the reduction can be carried out in polynomial time.  To complete the proof, we show that $\cI$ is 1-$\Pi_5$-satisfiable if and only if $\cI'$ is 2-$\Pi_9$-satisfiable. 

First, suppose that $\cI$ is 1-$\Pi_5$-satisfiable. Then there exists a linear ordering $\alpha$ on $V$ such that each constraint in $\cC$ is $\Pi_5$-satisfied by $\alpha$. \blue{Let} $\alpha'$ be a linear ordering of $V'$ obtained from $\alpha$ by preserving $\alpha$ and adding all elements in $V'-V$ such that, for each $C^i\in \cC$, the constraints in $\cC_{\Pi_9}(\gamma^i)$ are $\Pi_9$-satisfied. More precisely, if $\alpha(v_1^i)<\alpha(v_2^i)<\alpha(v_3^i)$, then $$\alpha'(1^i)<\alpha'(2^i)<\alpha'(3^i)<\alpha'(v_1^i)<\alpha'(5^i)<\alpha'(v_2^i)<\alpha'(v_3^i)$$ and, if $\alpha(v_3^i)<\alpha(v_2^i)<\alpha(v_1^i)$, then $$\alpha'(v_3^i)<\alpha'(v_2^i)<\alpha'(5^i)<\alpha'(v_1^i)<\alpha'(3^i)<\alpha'(2^i)<\alpha'(1^i).$$ Similarly, let $\beta'$ be a linear ordering of $V'$ obtained from $\alpha$ by preserving $\alpha$ and adding all elements in $V'-V$ such that, for each $C^i\in \cC$, the constraints in $\cC_{\Pi_9}(\delta^i)$ are $\Pi_9$-satisfied. More precisely, if $\alpha(v_1^i)<\alpha(v_2^i)<\alpha(v_
3^i)$, then $$\beta'(v_1^i)<\beta'(v_2^i)<\beta'(1^i)<\beta'(3^i)<\beta'(v_3^i)<\beta'(5^i)<\beta'(2^i)$$ and, if $\alpha(v_3^i)<\alpha(v_2^i)<\alpha(v_1^i)$, then $$\beta'(2^i)<\beta'(5^i)<\beta'(v_3^i)<\beta'(3^i)<\beta'(1^i)<\beta'(v_2^i)<\beta'(v_1^i).$$ By repeated applications of Lemma~\ref{l:2-Pi_9-caterpillar}, it now follows that each constraint in $\cC'$ is $\Pi_9$-satisfied by $\alpha'$ or $\beta'$ and, thus, $\cI'$ is 2-$\Pi_9$-satisfiable.

Second, suppose that $\cI'$ is 2-$\Pi_9$-satisfiable. Then there exist two linear orderings $\alpha'$ and $\beta'$ on $V'$ such that each constraint in $\cC'$ is $\Pi_9$-satisfied by $\alpha'$ or $\beta'$. It follows from Lemma~\ref{l:2-Pi_9-caterpillar} that, \blue{each solution to the instance $(V^i,\cC_{\Pi_9}(\gamma^i)\cup\cC_{\Pi_9}(\delta^i))$
\steven{consists of one element from $\{\gamma^i, \bar \gamma^i\}$ and
one element from $\{\delta^i, \bar \delta^i\}$}.
%consists of two elements in $\{\gamma^i,\delta^i,\bar\gamma^i,\bar\delta^i\}$.
We will show that the result holds if $\alpha'$ preserves $\gamma^i$ and if that $\beta'$ preserves $\delta^i$, noting that}
\steven{symmetrical arguments hold as long as one linear order preserves 
one element from
$\{\gamma^i, \bar \gamma^i\}$ and
the other linear order preserves one element from $\{\delta^i, \bar \delta^i\}$.}
% an analogous argument can be used to show that $\cI$ is 1-$\Pi_5$-satisfiable if $\alpha'$ and %$\beta'$ do not preserve $\gamma^i$ and $\delta^i$, respectively, but any other combination %of elements in $\{\gamma^i,\delta^i,\bar\gamma^i,\bar\delta^i\}$.}
Let $\alpha$ be the linear ordering $\alpha'-(V'-V)$ on $V$. Since, for each $C^i\in\cC$, we have $\alpha'(v_1^i)<\alpha'(v_2^i)<\alpha'(v_3^i)$ or $\alpha'(v_3^i)<\alpha'(v_2^i)<\alpha'(v_1^i)$, it is easily checked that $C^i$ is $\Pi_5$ satisfied by $\alpha$ and thus, $\cI$ is 1-$\Pi_5$-satisfiable.

The theorem now follows by combining both cases.
\end{proof}

%%%%%%%%%%%%%%%%%%%%%%%%%%%%%%%%%%%%%%%%%%%%%%%%%%%%%%%%%%%%%%%%

\section{Three linear orderings and an extension to phylogenetic trees}\label{sec:three}
In the last section, we have shown that 2-$\Pi_1$ is NP-complete. In this section we first show that the problem remains computationally hard if we allow for three instead of only two linear orderings. More formally, we show that  3-{\sc Caterpillar Compatibility} is NP-complete and, hence, by Observation~\ref{obs}, 3-$\Pi_1$ is NP-complete as well.\\

In the second part of the section, we extend the hardness result of 3-{\sc Caterpillar Compatibility} and show that the following decision problem is NP-complete.\\

\noindent 3-{\sc Tree Compatibility}

\noindent {\bf Instance.} A set $\cR$ of rooted triplets.

\noindent {\bf Question.} Do there exist at most three rooted phylogenetic trees such that each element in $\cR$ is displayed by at least one such tree?\\

In the flavor of Section~\ref{sec:two}, we start with a lemma whose proof is computational.
\steven{It shows that the three caterpillars shown in Figure \ref{fig:caterpillars} are
uniquely defined by their induced triplets, not just in the space of caterpillars but also in the
wider space of phylogenetic trees.}
 
%Throughout this section we use represent each $\Pi_1$-constraint $(a,b,c)$ as the corresponding rooted triplet $bc|a$ (see Section~\ref{sec:order2phylo}).

\begin{lemma} \label{lem:uniqueness}
 Let $C_1$, $C_2$, and $C_3$ be the three caterpillars that are shown in Figure~\ref{fig:caterpillars}, and let $\cC$ be the set of all triplets that are displayed by $C_1$, $C_2$, or $C_3$. Then, for each set of three trees, say $T_1$, $T_2$, and $T_3$, that is a solution to {\em 3}-{\sc Tree Compatibility} with input $\cC$, we have $C_1=T_1|\{0,1,\ldots,5\}$, $C_2=T_2|\{0,1,\ldots,5\}$, and $C_3=T_3|\{0,1,\ldots,5\}$.
 \end{lemma}

\begin{proof}
Computational proof (see appendix).
\end{proof}

\begin{figure}[h]
 \centering
  \includegraphics{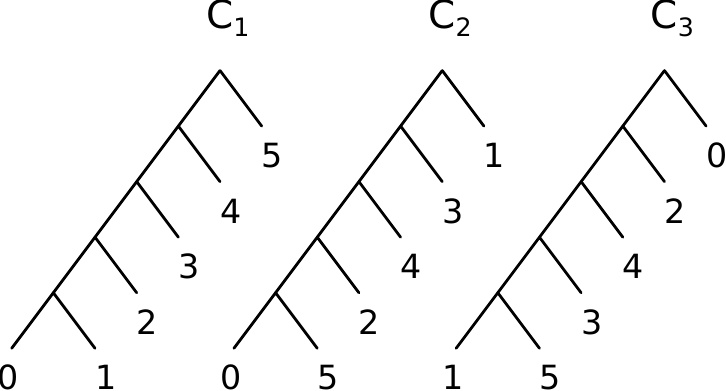}
  \caption{The three trees that are used in the proofs of Lemma~\ref{lem:uniqueness}, and Theorems~\ref{t:3caterpillar} and~\ref{t:3tree}.}
 \label{fig:caterpillars}
\end{figure}

To prove the first main result of this section (Theorem~\ref{t:3caterpillar}), we need a new definition. Let $e$ and $f$ be two leaves of a caterpillar $C$. If there exists a directed path from the parent of $f$ to $e$ in $C$, we say that $e$ is \textit{below} $f$ or, equivalently, $f$ is \textit{above} $e$ in $D$ and write $e \prec f$.

\begin{theorem}\label{t:3caterpillar}
The problems {\em 3-}{\sc Caterpillar Compatibility} and, hence,  {\em 3}-$\Pi_1$ are NP-complete.
\end{theorem}

\begin{proof}
Trivially, 3-{\sc Caterpillar Compatibility} is  in NP. Let $\cI$ be an instance of 2-{\sc Caterpillar Compatibility}, i.e. $\cI$ is a set of triplets. To show that the theorem holds, we reduce $\cI$ to an instance of 3-{\sc Caterpillar Compatibility}. Let $\cC$ be the set of triplets as defined in the statement of Lemma~\ref{lem:uniqueness}. Furthermore, we define the following five sets of triplets, where, for each triplet $ab|c\in\cI$, $\ab$ denotes a new taxon that is not a leaf label of a triplet in $\cI$.

\begin{align*}
 \cR_1 &= \bigcup_{ab|c\in\cI}\{3a|5, 3a|1, 4a|0,3b|5, 3b|1, 4b|0,3c|5, 3c|1, 4c|0,3\ab|5, 3\ab|1, 4\ab|0\},\\
 \cR_2 &= \bigcup_{ab|c\in\cI}\{4\ab|2\},\\
 \cR_3 &= \bigcup_{ab|c\in\cI}\{0\ab|c\},\\
 \cR_4 &= \bigcup_{ab|c\in\cI}\{25|a, 12|a, 0a|2, 25|b, 12|b, 0b|2\},\textnormal{ and}\\
 \cR_5 &= \bigcup_{ab|c\in\cI}\{5a|\ab, 5b|\ab, 1a|\ab, 1b|\ab\}.
\end{align*}
Now, let $\cR=\cR_1\cup\cR_2\cup\ldots\cup\cR_5$. Clearly, the number of elements in $\cC$ and $\cR$ is polynomial in 6 and $|\cI|$, respectively. To complete the proof, we show that $\cI$ is a `yes'-instance of 2-{\sc Caterpillar Compatibility} if and only if \steven{$\cC\cup\cR$} is a `yes'-instance of 3-{\sc Caterpillar Compatibility}. 

First, suppose that $\cI$ is a `yes'-instance of 2-{\sc Caterpillar Compatibility}. Then, there exist two caterpillars  $S_1$ and $S_2$ such that each triplet in $\cI$ is displayed by $S_1$ or $S_2$. Illustrated in Figure~\ref{fig:solution}, we obtain three new caterpillars $T_1$, $T_2$, and $T_3$ as follows. For $i\in\{1,2,3\}$, start by setting $T_i=C_i$, where $C_i$ is the  caterpillar shown in Figure~\ref{fig:caterpillars} and insert $S_1$ into $T_1$ below 2 and above 1, and $S_2$ into $T_2$ below 2 and above 5. It is easily checked that the resulting trees display all  triplets in \steven{$\cC$}. Furthermore, for each triplet $ab|c \in \cI$ that is displayed by $S_1$, add the corresponding $\ab$ taxon to $T_1$ such that $\ab$ is above \steven{$a$ and $b$} and below $c$ and add $\ab$ to $T_2$ such that $\ab$ is above $a$ and $b$, and below 2. Similarly, for each triplet $ab|c \in \cI$ that is not displayed by $S_1$, add the corresponding $\ab$ taxon to $T_2$ such that $\ab$ is is above \steven{$a$ and $b$} and below $c$ and add $\ab$ to $T_1$ such that $\ab$ is above $a$ and $b$, and below 2.
Lastly, for each triplet $ab|c\in \cI$, add each taxon in $\{a,b,\ab\}$ to $T_3$ such that $4\prec\ab\prec 2\prec a$, $2\prec a\prec 0$, and $2\prec b\prec 0$ holds. We next show that each triplet in $\cR$ is displayed by at least one tree of $T_1$, $T_2$, and $T_3$. For each $x\in\{a,b,c,\ab\}$, the triplet $3x|5$ is displayed by $T_1$, the triplet   $3x|1$ is displayed by $T_2$, and the triplet  $4x|0$ is displayed by $T_3$. Hence, each triplet in $\cR_1$ is displayed by $T_1$, $T_2$, or $T_3$. Furthermore, each triplet in $\cR_2$ is displayed by $T_3$ and, depending on whether or not a triplet $ab|c\in\cI$ is displayed by $S_1$, the corresponding triplet in $\cR_3$ is displayed by $T_1$ or $T_2$. Turning, to the triplets in $\cR_4$, a straightforward check shows that, for each $ab|c\in\cI$, the two corresponding  triplets $0a|2$ and $0b|2$ in $\cR_4$ are displayed by $T_1$ (and $T_2$), while the remaining four triplets in $\cR_4$ are displayed by $T_3$. Lastly, for each triplet $ab|c\in\cI$, the first two corresponding triplets in $\cR_5$ are displayed by $T_2$ while the last two corresponding triplets in $\cR_5$ are displayed by $T_1$. Hence, each triplet in \steven{$\cC\cup\cR$} is displayed by one of $T_1$, $T_2$, or $T_3$ and, therefore, \steven{$\cC\cup\cR$} is a `yes'-instance of 3-{\sc Caterpillar Compatibility}.

\begin{figure}[h]
 \centering
 \includegraphics[width=11cm]{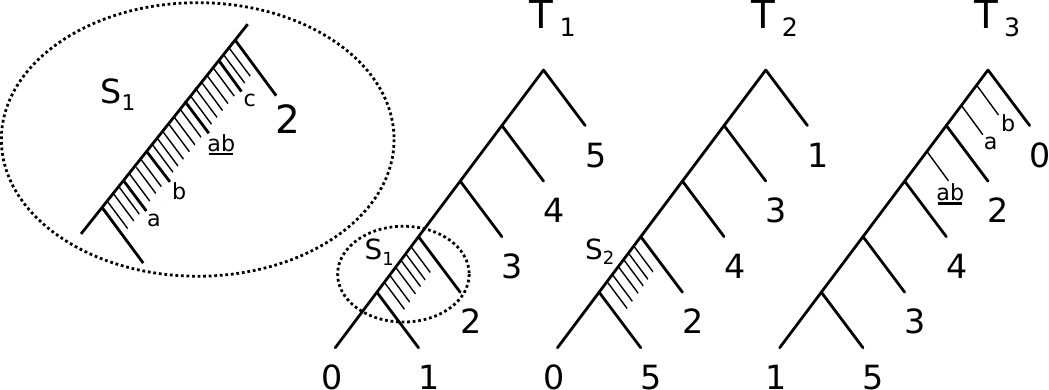}
 \caption{Right: A solution $T_1$, $T_2$, and $T_3$ to the transformed instance of 3-{\sc Caterpillar Compatibility}, given a solution $S_1$ and $S_2$ to some 2-{\sc Caterpillar Compatibility} instance. Left: A more detailed view of the part of the caterpillar $T_1$ inside the dotted circle. }
 \label{fig:solution}
\end{figure}

% \marginpar{For Figure~\ref{fig:solution}, please replace $D_i$ with $T_i$.}

Second, suppose that \steven{$\cC\cup\cR$} is a `yes'-instance of 3-{\sc Caterpillar Compatibility}. Then, there exist three caterpillars  $T_1$, $T_2$,  and $T_3$ such that each triplet in \steven{$\cC\cup\cR$} is displayed by $T_1$, $T_2$,  or $T_3$. By Lemma~\ref{lem:uniqueness}, we may assume without loss of generality that $C_1=T_1|\{0,1,\ldots,5\}$, $C_2=T_2|\{0,1,\ldots,5\}$, and $C_3=T_3|\{0,1,\ldots,5\}$, where $C_1$, $C_2$, and $C_3$ are the three caterpillars that are shown in Figure~\ref{fig:caterpillars}. We make three observations that follow from the different triplet sets $\cR_1$, $\cR_2$, and $\cR_3$. Let $ab|c\in\cI$. 
\begin{enumerate}
\item [(1)]  For each $x\in\{a,b,c,\ab\}$, the triplet $3x|5$ is only displayed by $T_1$, the triplet $3x|1$ is only displayed by $T_2$, and the triplet  $4x|0$ is only displayed by $T_3$. In particular the root of $T_1$ is the parent of 5 and, similarly, the root of $T_2$ (resp. $T_3$) is the parent of 1 (resp. 0). 
\item [(2)] The triplet $4\ab|2$ is only displayed by $T_3$ and, as a consequence, $\ab$ is below 2 in $T_3$.
\item [(3)] By Observation (1), the triplet $0\ab|c$ is not displayed by $T_3$ and, hence,  $\ab$ is below  $c$ in $T_1$ or $T_2$.
\end{enumerate}
Now consider the triplets in $\cR_4$. We claim that, for each triplet $ab|c\in\cI$, the two taxa $a$ and $b$ are above 2 in $T_3$. Assume that $a$ is not above 2 in $T_3$.  Then, by Observation (1), $25|a$ is only displayed by $T_2$ and $12|a$ is only displayed by $T_1$ and, therefore, $a$ is above 2 in both of $T_1$ and $T_2$. Hence, $0a|2$ is not displayed by $T_1$ or $T_2$ and, by Observation (1), certainly not by $T_3$; a contradiction. Similarly, assume that $b$ is not above 2 in $T_3$. Then, by Observation (1), $25|b$ is only displayed by $T_2$ and $12|b$ is only displayed by $T_1$ and, therefore, $b$ is above 2 in both $T_1$ and $T_2$. Hence, $0b|2$ is not displayed by $T_1$ or $T_2$ and, by Observation (1), certainly not by $T_3$; a contradiction. Now, since $a$ and $b$ are both above 2 in $T_3$, it follows from Observation (2) that $a$ and $b$ are both above $\ab$ in $T_3$. In turn, this implies that $T_3$ does not display a triplet $ya|\ab$ or $yb|\ab$  with $y\in\{1,5\}$. 
Finally consider the triplets in $\cR_5$. For each triplet $ab|c\in\cI$, the triplets $5a|\ab$ and $5b|\ab$ are only displayed by  $T_2$ due to Observation (1) and the previous claim. Similarly, the triplets $1a|\ab$ and $1b|\ab$ are only displayed by $T_1$ due to Observation (1) and the previous claim. Hence, $\ab$ is above $a$ and $b$ in both $T_1$ and $T_2$. Now, recall that $\ab$ is below $c$ in at least one of $T_1$ or $T_2$ by Observation (3). It follows that each $ab|c\in\cI$ is displayed by $T_1$ or $T_2$.  and, therefore, $\cI$ is a `yes'-instance of 2-{\sc Caterpillar Compatibility}. By combining both cases, it follows that 3-{\sc Caterpillar Compatibility} is NP-complete and, hence, by Observation~\ref{obs}, 3-$\Pi_1$ is also NP-complete.
\end{proof}

For a given set of triplets, we next show that it is not only NP-complete to decide if there exist three caterpillars such that each element in $\cR$ is displayed by at least one of these caterpillars, but also NP-complete to decide if there exist three (arbitrary) rooted phylogenetic trees such that each element in $\cR$ is displayed by at least one of these trees. We start with a few new definitions.

Let $T$ be a caterpillar. Furthermore, let $\{c,c'\}$ be the unique cherry in $T$ and let $x$ be the leaf of $T$ such that the directed path from the root of $T$ to $x$ contains precisely one edge. We refer to the directed path from the root of $T$ to the parent of $c$ as the {\it spine} of $T$ and to each other edge in $T$ as a {\it leg} of $T$. Note that the definitions of spine and leg naturally carry over to each tree obtained from $T$ by \steven{subdividing edges} of $T$ except for the edges directed into $c$, $c'$, and $x$ respectively. We call such a tree a {\it relaxed caterpillar}. Lastly, a rooted subtree of a rooted phylogenetic tree $T$ \steven{is \emph{pendant}} if it can be detached from $\cT$ by deleting a single edge. Note that each leaf of $T$ is a pendant subtree of $T$.

\begin{theorem}\label{t:3tree}
The problem {\em 3-}{\sc Tree Compatibility} is NP-complete.
\end{theorem}

\begin{proof}
The proof of this theorem is similar to that of Theorem~\ref{t:3caterpillar}.
Trivially, 3-{\sc Tree Compatibility} is  in NP. Let $\cI$ be an instance of 2-{\sc Caterpillar Compatibility}, i.e. $\cI$ is a set of triplets. To show that the theorem holds, we reduce $\cI$ to an instance of 3-{\sc Tree Compatibility}. Let $\cC$ be the set of triplets as defined in the statement of Lemma~\ref{lem:uniqueness}. Furthermore, we define the following six sets of triplets, where, for each triplet $ab|c\in\cI$, $\ab$ denotes a new taxon that is not a leaf label of a triplet in $\cI$.

\begin{align*}
 \cR_1 &= \bigcup_{ab|c\in\cI}\mbox{  }\bigcup_{x\in\{a,b,c,\ab\}}\{2x|5, 3x|5, 4x|5, 2x|1, 3x|1, 4x|1, 2x|0, 3x|0, 4x|0\},\\
 \cR_2 &= \bigcup_{ab|c\in\cI}\{01|a, 05|a,01|b, 05|b,01|c, 05|c,01|\ab, 05|\ab\},\\
 \cR_3 &= \bigcup_{ab|c\in\cI}\{3\ab|2, 4\ab|2, 5\ab|2\},\\
 \cR_4 &= \bigcup_{ab|c\in\cI}\{0\ab|c\},\\
 \cR_5 &= \bigcup_{ab|c\in\cI}\{25|a, 12|a, 0a|2, 25|b, 12|b, 0b|2\}, \textnormal{ and} \\
 \cR_6 &= \bigcup_{ab|c\in\cI}\{5a|\ab, 5b|\ab, 1a|\ab, 1b|\ab\}.
\end{align*}
Note that only $\cR_1$, $\cR_2$, and $\cR_3$ contain triplets that are not used in the reduction presented in the proof of Theorem~\ref{t:3caterpillar}.
Now, let $\cR=\cR_1\cup\cR_2\cup\ldots\cup\cR_6$. Clearly, the number of elements in $\cC$ and $\cR$ is polynomial in 6 and $|\cI|$, respectively. To complete the proof, we show that $\cI$ is a `yes'-instance of 2-{\sc Caterpillar Compatibility} if and only if \steven{$\cC\cup\cR$} is a `yes'-instance of 3-{\sc Tree Compatibility}. 

First, suppose that $\cI$ is a `yes'-instance of 2-{\sc Caterpillar Compatibility}. Then, there exist two caterpillars  $S_1$ and $S_2$ such that each triplet in $\cI$ is displayed by $S_1$ or $S_2$. Let $T_1$, $T_2$, and $T_3$ be the same phylogenetic trees (caterpillars) reconstructed from $S_1$ and $S_2$ as in the proof of Theorem~\ref{t:3caterpillar}. Building on this proof, it remains to show that all triplets in $\cR_1$, $\cR_2$, and $\cR_3$ are displayed by at least one of $T_1$, $T_2$, and $T_3$. For each $x\in\{a,b,c,\ab\}$, the triplets $2x|5$, $3x|5$, and $4x|5$ are displayed by $T_1$, the triplets $2x|1$, $3x|1$, and $4x|1$ are displayed by $T_2$, and the triplets $2x|0$, $3x|0$, $4x|0$ are displayed by $T_3$. Hence, each triplet in $\cR_1$ is displayed by $T_1$, $T_2$, or $T_3$. Furthermore, turning to the triplets in $\cR_2$, each triplet $01|x$ is displayed by $T_1$ and each triplet $05|x$ is displayed by $T_2$. Lastly, each triplet in $\cR_3$ is displayed by $T_3$. In conclusion, each triplet in \steven{$\cC\cup\cR$} is displayed by one of $T_1$, $T_2$, or $T_3$ and, therefore, \steven{$\cC\cup\cR$} is a `yes'-instance of 3-{\sc Tree Compatibility}.

Second, suppose that \steven{$\cC\cup\cR$} is a `yes'-instance of 3-{\sc Tree Compatibility}. Then, there exist three trees  $T_1$, $T_2$,  and $T_3$ such that each triplet in \steven{$\cC\cup\cR$} is displayed by $T_1$, $T_2$,  or $T_3$. Again, by Lemma~\ref{lem:uniqueness}, we may assume without loss of generality that $C_1=T_1|\{0,1,\ldots,5\}$, $C_2=T_2|\{0,1,\ldots,5\}$, and $C_3=T_3|\{0,1,\ldots,5\}$, where $C_1$, $C_2$, and $C_3$ are the three caterpillars that are shown in Figure~\ref{fig:caterpillars}. For the remainder of the  proof, we relax the definition of `above' (resp. `below') as defined prior to the proof of Theorem~\ref{t:3caterpillar}. In particular, for two leaves $x$ and $y$ in $L(T_i)$ with $i\in\{1,2,3\}$, we say that $x$ is {\it below} $y$ or, equivalently, $y$ is {\it above} $x$ if there is a directed path from $\lca_{T_i}(c,y)$ to $\lca_{T_i}(c,x)$ that contains at least one edge and where $c$ is a leaf of the cherry in $C_i$. We next consider the triplets in $\cR$. Let $ab|c\in\cI$. 
\begin{enumerate}
\item [(1)]  We claim that the root of $T_1$ is the parent of 5. To see that this is indeed true, consider the triplets $2x|5, 3x|5, 4x|5$ with $x\in\{a,b,c,\ab\}$. For $x$ being fixed, no two of these three triplets can  simultaneously be displayed by $T_2$ or $T_3$. Hence, at least one such triplet is only displayed by $T_1$; thereby implying the correctness of the claim. A similar argument can be used to show that the parent of 1 is the root of $T_2$ (by exploiting the triplets $2x|1$, $3x|1$, and $4x|1$) and that the parent of 0 is the root of $T_3$ (by exploiting the triplets $2x|0$, $3x|0$, and $4x|0$). 
\item[(2)] By Observation (1), the triplets in $\cR_2$ guarantee that, for each $x\in\{a,b,c,\ab\}$, the triplet $01|x$ is only displayed by $T_1$ and the triplet $05|x$ is only displayed by $T_2$. In other words $\{0,1\}$ is a cherry of $T_1$ and $\{0,5\}$ is a cherry of $T_2$. 
\item [(3)] At least one of the three corresponding triplets in $\cR_3$ is only displayed by $T_3$ and, hence, $\ab$ is below 2 in $T_3$.
\item [(4)] By Observation (1), the triplet $0\ab|c$ in $\cR_4$ is not displayed by $T_3$ and, hence,  $\ab$ is below  $c$ in $T_1$ or $T_2$.
\item [(5)] Considering $\cR_5$ and using the same argument as in the proof of Theorem~\ref{t:3caterpillar}, it follows that $a$ and $b$ are above 2 in $T_3$.
\item[(6)] By Observations (3) and (5), no triplet in $\cR_6$ is displayed by $T_3$. Moreover, due to Observation (1), we have that $5a|\ab$ and $5b|\ab$ are only displayed by $T_2$ and that $1a|\ab$ and $1b|\ab$ are only displayed by $T_1$; thereby implying that $\ab$ is above $a$ and $b$ in $T_1$ and $T_2$.
\end{enumerate}
Now, by combining Observations (2), (4) and (6), it now follows that each $ab|c\in\cI$ is displayed by $T_1$ or $T_2$. 

Guided by $T_1$ and $T_2$, we next construct two caterpillars $S_1$ and $S_2$. For $i\in\{1,2\}$, let $C_i'$ be the relaxed caterpillar $T_i(\{0,1,2,3,4,5\})$, and let $c_i$ be a leaf of the cherry in $C_i'$. By Observations (1) and (2), note that $C_i'$ is indeed a relaxed caterpillar. Now, let $S_i=T_i$. For each maximum-size pendant subtree $S$ of $S_i$ whose leaf set $\{x_1,x_2,\ldots,x_n\}$ is a subset of $L(S_i)-\{0,1,2,3,4,5\}$, repeat the following (illustrated in Figure~\ref{fig:tree-cat}) until the resulting trees are both caterpillars. If $S$ can be detached from $S_i$ by deleting an edge $(u_1,v)$ such that $u_1$  corresponds to a degree-2 vertex on the spine of $C_i'$, then delete $(u_1,v)$, subdivide the edge directed into $u_1$ with $n-1$ new vertices $u_2,u_3,\ldots, u_{n}$, and, for each $j\in\{1,2,\ldots,n\}$, add a new edge $(u_j,x_j)$. Otherwise, $S$ can be detached from $S_i$ by deleting an edge $(u,v)$ such that $u$ corresponds to a degree-2 vertex on a leg of $C_i'$. Let $\ell$ be the unique element in $\{0,1,2,3,4,5\}$ such that there is a directed path from $u$ to $\ell$ in $S_i$. Then, delete $(u,v)$ and contract the resulting degree-2 vertex $u$, subdivide the edge directed into $\lca_{S_i}(c_i,\ell)$ with $n$ new vertices $u_1,u_2,\ldots, u_{n}$, and, for each $j\in\{1,2,\ldots,n\}$, add a new edge $(u_j,x_j)$. 
\begin{figure}[h]
 \centering
  \includegraphics{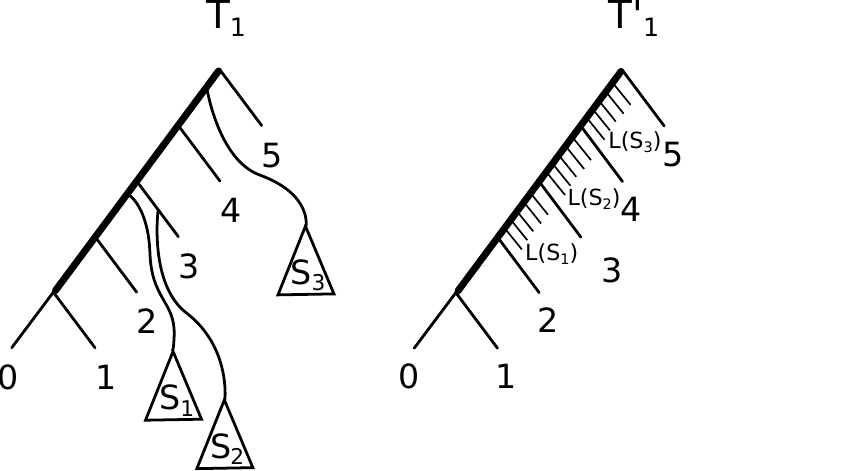}
 \caption{The construction of a caterpillar $T_1'$ from $T_1$ as described in the second part of the proof of Theorem~\ref{t:3tree}.}
 \label{fig:tree-cat}
\end{figure}

We complete the proof by showing that each triplet $ab|c\in\cI$ is displayed by $S_1$ or $S_2$. 
%\steven{To achieve this it is sufficient to show that $S_1, S_2, T_3$
%still display all triplets in $\cC \cup \cR$. In particular, as argued earlier in the proof,
%such instances have the property that each $ab|c \in \cI$ is displayed by $S_1$ or $S_2$.} 
Intuitively, the transformation \steven{from trees into caterpillars} is safe because, for any two leaves in $L(T_i)$ with $i\in\{1,2\}$ for which the `above'-relationship holds in $T_i$, the `above'-relationship also holds in $S_i$. 

%\steven{First, observe that all triplets in $\cC$ are still displayed after the
%transformation, since taxa $\{0,1,2,3,4,5\}$ are untouched.  It remains only to consider %triplets in $\cR$. By Observation (2), all triplets
%of the form $01|x$ are displayed by $T_1$, and all triplets of the form $05|x$ are
%displayed by $T_2$. These constitute all triplets from $\cR_2$. In $T_1$ the transformation %does not move any taxa below $\lca_{T_1}(0,1)$, and in $T_2$ the transformation does not %move any
%taxa below $\lca_{T_2}(0,5)$, so triplets in $\cR_2$ are preserved. Next, consider $\cR_1$.
%We know from Observation (1) that $T_1$ displays at least one of $2x|5, 3x|5, 4x|5$, for
%every $x \in \{a,b,c, \ab\}$. The transformation moves $x$ onto the spine, ensuring
%that $S_1$ displays all $2x|5, 3x|5, 4x|5$. Similarly...
%}

To ease reading in this part of the proof, we pause and introduce new terminology. Let $S$ be a subtree of $T_i$ with $L(S)\subseteq L(T_i)-\{0,1,2,3,4,5\}$. If $S$ can be detached from $T_i$ by deleting an edge $(u,v)$ such that $u$  corresponds to a degree-2 vertex on the spine of $C_i'$, we call $S$ a {\it spine subtree} of $T_i$. Otherwise, we call $S$, an {\it $\ell$-leg subtree} of $T_i$, where $\ell$ is the unique element in $\{0,1,2,3,4,5\}$ such that there exists a directed path from $u$ to $\ell$ in $T_i$.
Now, let $ab|c\in\cI$.

\steven{First, assume that $0\ab|c$ is displayed by $T_1$. (Then, because $1a|\ab$ and
$1b|\ab$ are definitely displayed by $T_1$, $ab|c$ is displayed by $T_1$.)}
Since $T_1$ displays $0\ab|c$, $1a|\ab$ and $1b|\ab$, it follows that 
%\steven{it is not possible for two or more elements from $\{a,b,c,\ab\}$ to be in the same
%spine or leg subtree}.
there does not exist a spine \steven{or leg} subtree $S$ in $T_1$ with $\{a,b,c\}\subseteq L(S)$. \steven{Clearly, $a$ and $c$ cannot be together in a leg or spine subtree without
$b$, and similarly $b$ and $c$ cannot be together without $a$, because this would contradict
the fact that $T_1$ displays $ab|c$. So suppose $a$ and
$b$ are together without $c$ in a leg or spine subtree. Due to the fact that
$0\ab|c$, $1a|\ab$ and $1b|\ab$ are displayed by $T_1$, the transformation ensures that in $S_1$ there is a directed path from the parent of $c$ to $\lca_{T_1}(a,b)$ i.e. that $ab|c$ is
displayed. Let us then consider the remaining case when $a$, $b$ and $c$ are in three distinct
subtrees (spine or leg). Observe that it is not possible for all three distinct subtrees to
be $\ell$-leg subtrees, for the same $\ell$. This, again, is because of the three
triplets $0\ab|c$, $1a|\ab$ and $1b|\ab$. Hence, at most $a$ and $b$ can be in distinct
$\ell$-leg subtrees, for the same $\ell$. Under such circumstances the transformation
again ensures that $ab|c$ will be displayed by $S_1$.}

%Hence, if each of $a$, $b$, and $c$ is contained in a spine subtree in $T_1$, it follows by the %construction of $S_1$ from $T_1$ that, as $T_1$ displays $ab|c$, so does $S_1$. Now assume %that not each of $a$, $b$, and $c$ is contained in a spine subtree in $T_1$. If there exists at %most two elements in $\{a,b,c\}$ that are not contained in a spine subtree of $T_1$, then, as %$T_1$ displays $ab|c$, $S_1$ also displays $ab|c$. We may therefore assume that no element %in $\{a,b,c\}$ is contained in a spine subtree. If $a$ is contained in an $\ell$-leg subtree of %$T_1$ and at least one of $b$ and $c$ is contained in an $\ell'$-leg subtree of $T_1$ with %$\ell\ne\ell'$, then it follows again, by construction of $S_1$, that $S_1$ displays $ab|c$. %Lastly, assume that there exists an $\ell$-leg subtree $S$ of $T_1$ such that %$\{a,b,c\}\subseteq L(S)$. Then, at least one of the three triplets $0\ab|c$, $1a|\ab$ and %$1b|\ab$ is not displayed by $T_1$; a contradiction. It now follows that $ab|c$ is displayed by %$S_1$.

Second, assume that \steven{$0\ab|c$} is displayed by $T_2$ \steven{(and hence
$ab|c$ is displayed by $T_2$)}. By replacing the triplets $1a|\ab$ and $1b|\ab$ with $5a|\ab$ and $5b|\ab$, respectively, we can use an analogous argument to show that $ab|c$ is displayed by $S_2$. In conclusion, each triplet $ab|c\in\cI$ is displayed by the caterpillar $S_1$ or $S_2$ and, thus, $\cI$ is a `yes'-instance of 2-{\sc Caterpillar Compatibility}.

By combining both cases, and noting that the construction described in the previous paragraph of this proof can be done in polynomial time, it follows that 3-{\sc Tree Compatibility} is NP-complete.
\end{proof}

%%%%%%%%%%%%%%%%%%%%%%%%%%%%%%%%%%%%%%%%%%%%%%%%%%%%%%%%%%%%%%%%%%%%%%%%%%

\section{A constant number of trees is not enough} 
\label{sec:notenough}

We begin this section with several new definitions. Let $T$ be an unrooted \steven{binary} phylogenetic tree, and let $T_r$ be a rooted binary phylogenetic tree.  We say that~$T_r$ is a \emph{rooting} of~$T$ if~$T_r$ can be obtained from~$T$ by subdividing an edge~$e$ of~$T$ by a new vertex~$\rho$, \blue{which is regarded as the root}, and directing all edges away from~$\rho$. The edge~$e$ is then called the \emph{root location} of~$T_r$ in~$T$.
The next definition introduces a special set of rooted triplets that will play an important role throughout this section. Let~$\cT_n$ denote the full set of triplets over~$n$ leaves, i.e.

\[
\cT_n = \{ ab|c \quad:\quad a,b,c\in \n,\quad a\neq b \neq c\neq a\}.
\]

Furthermore, we use $\tau(n)$ to denote the minimum number of binary rooted phylogenetic trees of order~$n$, such that each triplet in~$\cT_n$ is displayed by at least one of these trees.

We will show that $\tau(n)\rightarrow\infty$ when~$n\rightarrow\infty$. In other words, we show that a constant number of trees does not suffice to display all possible triplet sets.

To prove this, we will use the following theorem shown by Martin and Thatte \cite{martin2013maximum}, which improves an earlier bound by Sz\'ekely and Steel \cite{steel2009improved}. Let~$\umast_k(n)$ denote the smallest number~$L$ such that, for any collection~$\cC$ of~$k$ binary unrooted phylogenetic trees of order~$n$, there exists a binary unrooted phylogenetic tree~$T^*$ of order~$L$ that is displayed by each tree in~$\cC$.

\begin{theorem}[Martin and Thatte]\label{thm:thatte}
For some constant~$c>0$, $$\umast_2(n) > c\sqrt{\log(n)}.$$
\end{theorem}

To use the above theorem, we first extend it to more than~2 trees.

\begin{lemma}\label{lem:umast}
For some constant~$c>0$ and~$k > 2$, $$\umast_k(n) > c\sqrt{\log(\umast_{k-1}(n))}.$$
\end{lemma}
\begin{proof}
Let~$k\geq 3$. Consider a collection~$\cC$ of~$k$ unrooted phylogenetic trees of order~$n$. Consider an arbitrary $k-1$-sized subset~$\cC'$ of~$\cC$ and let~$T^*$ be the only tree in~$\cC\setminus\cC'$. Let~$T'$ be an unrooted phylogenetic tree that is displayed by each tree in~$\cC'$ and has order $\umast_{k-1}(n)$ (i.e. it has a maximum number of leaves). Let~$T^{**}$ be the unique subtree of~$T^*$ that has the same leaf set as~$T'$. By Theorem~\ref{thm:thatte}, there exists an unrooted phylogenetic tree~$T^{'**}$ that is displayed by both~$T'$ and~$T^{**}$ and has order $\umast_2(\umast_{k-1}(n)) > c\sqrt{\log(\umast_{k-1}(n))}$. Since~$T^{'**}$ is displayed by each tree in~$\cC$, the lemma follows.
\end{proof}

The following lemma follows directly from Theorem~\ref{thm:thatte} and Lemma~\ref{lem:umast} by induction on~$k$.

\begin{lemma}\label{lem:umastinf}
For all~$k\in N$, $\umast_k(n)\rightarrow\infty$ when $n\rightarrow\infty$.
\end{lemma}

We now make the step from unrooted to rooted trees by proving the following lemma.

\begin{lemma}\label{lem:missingtriplet}
Let~$T$ be a binary unrooted phylogenetic tree of order~\steven{$n \geq 4$} and let~$T_1,\ldots ,T_k$ be~$k$ rootings of~$T$. If $n > \bound$, then there exists a triplet in~$\cT_n$ that is not displayed by any of~$T_1,\ldots ,T_k$.
\end{lemma}
\begin{proof}
Let~\steven{$c \geq 2$} be the number of cherries (i.e. pairs of leaves with a common neighbour) of~$T$. Let~$\ell$ be the length of the longest chain (i.e. a path of which each vertex is adjacent to exactly one leaf) of~$T$.

We claim that~$T$ has at most~$2c-3$ chains. To see this, consider the tree~$T^*$ obtained from~$T$ by deleting all leaves and subsequently suppressing all degree-2 vertices. Then,~$T^*$ has~$c$ leaves and hence~$2c-3$ edges. Each chain of~$T$ corresponds to an edge of~$T^*$ and hence~$T$ has at most~$2c-3$ chains.

First assume that~$k < 2c$. Then there exists some cherry $\{a,b\}$ such that at most one of the edges incident to~$a$ and~$b$ is a root location of at least one of the trees~$T_1,\ldots ,T_k$ in~$T$. Say that the edge incident to~$a$ is not a root location and let~$c$ be any leaf distinct from~$a$ and~$b$. Then it can easily be checked that the triplet~$bc|a$ is not displayed by any of~$T_1,\ldots ,T_k$.

Now assume that~$k\geq 2c$. Tree~$T$ has at most~2 leaves per cherry plus at most~$\ell$ leaves per chain, hence
\[
n \leq 2c + \ell(2c-3).
\]
It follows that
\[
\ell \geq \frac{n-2c}{2c-3} \geq \frac{n-k}{k-3}.
\]
Now, since~$n > \bound$, it follows that~$\ell > k+2$.

Consider some chain~$C$ of~$T$ with length
at least $k+3$.
%greater than~$3k$.
Since there are at most~$k$ edges that are root locations of the trees~$T_1,\ldots ,T_k$ in~$T$,  
there exists some
subpath~$(u,v,w)$ of~$C$, with $a,b,c$ being the leaves adjacent to $u,v,w$ respectively,
such that the edge from $v$ to $b$ is not a root location.
% such that the edges~$\{u,v\}$ and~$\{v,w\}$ are not root %locations. Let~$a,b,c$ be the %leaves adjacent to~$u,v,w$ respectively.
Then it can easily be checked that the triplet~$ac|b$ is  not displayed by any of~$T_1,\ldots ,T_k$.
\end{proof}

We can now prove the main theorem of this section.

\begin{theorem}\label{thm:infty}
$\tau(n)\rightarrow\infty$ when~$n\rightarrow\infty$
\end{theorem}
\begin{proof}
Suppose to the contrary that there exists some natural number~$K$ such that~$\tau(n)\leq K$ for all~$n\in\NN$. By Lemma~\ref{lem:umastinf}, there exists a natural number~$N\in\NN$ such that $\umast_K(N) > \boundcap$. Let~$\cC_r$ be a collection of~$K$ binary rooted phylogenetic trees of order~$N$ such that each triplet in~$\cT_N$ is displayed by at least one tree in~$\cC_r$. Such a collection exists by the assumption that~$\tau(N)\leq K$.

Let~$\cC_u$ be the collection of unrooted trees obtained from~$\cC_r$ by omitting the orientations of the edges and suppressing the vertices with degree~2 (the former roots). Then there exists some binary unrooted phylogenetic tree~$T_u$ that is displayed by each tree in~$\cC_u$ and has order~$\umast_K(N)$. Let~$L$ be the leaf-set of~$T_u$ and consider the set of rooted trees~$\cC_r|L$ obtained by restricting~$\cC_r$ to~$L$, i.e. for each tree~$T\in\cC_r$, the set~$\cC_r|L$ contains the tree obtained from the smallest subtree of~$T$ containing all elements of~$L$ by suppressing all indegree-1 outdegree-1 vertices. Observe that each tree in~$\cC_r|L$ is a rooting of~$T_u$ and has order $\umast_K(N)$. Hence, since~$\umast_K(N) > \boundcap$, it follows from Lemma~\ref{lem:missingtriplet} that there exists some triplet~$ab|c\in\cT_N$ that is not displayed by any tree in~$\cC_r|L$. Hence, $ab|c$ is not displayed by any tree in~$\cC_r$. This contradicts our assumption that each triplet in~$\cT_N$ is displayed by at 
least one tree in~$\cC_r$.
\end{proof}

Let $\tau_c(n)$ be the minimum number of binary rooted \emph{caterpillar} trees of order~$n$, such that each triplet in~$\cT_n$ is displayed by at least one of these trees. Clearly, $\tau(n) \leq \tau_c(n)$. We have the following logarithmic upper bound:

\begin{lemma}
\label{lem:log}
$\tau(n) \leq \tau_c(n) \leq \bigg \lceil \frac{ \log n(n-1)(n-2) - \log 2 }{\log (3/2)} \bigg \rceil$ 
\end{lemma}
\begin{proof}
It is well known that, given a set of triplets, it is possible to find in polynomial time a tree (in fact, a caterpillar) that is consistent with at least 1/3 of the triplets, see \cite{ByrkaEtAl2008} for a discussion. Observe that $|\cT_n| = 3\binom{n}{3}$ and that a triplet set containing only one triplet is trivially compatible. Combining these facts shows that $k$ caterpillars are sufficient
to display all triplets in $\cT_n$, where $k$ is the smallest integer that satisfies the following
inequality:
\[
3 \binom{n}{3} (2/3)^{k} \leq 1
\]
Rearranging for $k$ gives the desired result.
\end{proof}

The lower bound implicit in Theorem \ref{thm:infty} and the upper bound in Lemma \ref{lem:log} are very weak. To highlight this we computed $\tau(n)$ and $\tau_c(n)$ exactly for small
values of $n$ using Integer Linear Programming (ILP). We defer the details of the ILP to the appendix. The results are shown in the following table. For $n \geq 13$ the ILP for computation
of $\tau(n)$ did not terminate in reasonable time, but the slightly more constrainted ILP for computation of $\tau_c(n)$ did.

\begin{table}[h]
\begin{center}
\begin{tabular}{|c|c|c|c|c|c|c|c|c|c|c|c|}
\hline
$n$          & 3 & 4 & 5 & 6 & 7 & 8 & 9 & 10 & 11 & 12 & $13 \leq n \leq 20$ \\ \hline
\hline
$\tau(n)$    & 3 & 3 & 4 & 4 & 4 & 4 & 4 & 4  & 4  & 4  & $\leq 5$       \\ \hline
$\tau_c(n)$ & 3 & 3 & 4 & 4 & 4 & 4 & 4 & 4  & 4  & 4  & 5       \\ \hline
Lemma \ref{lem:log} bound      & 3 & 7 & 9 & 11 & 12 & 13 & 14 & 15  & 16  & 17  & 17-21       \\ \hline
\end{tabular}
\end{center}
\end{table}

The possibility thus remains that, for all $n$, $\tau(n) = \tau_c(n)$. Resolving this is an interesting open problem. We do however already know that for some triplet sets with fewer
than $3\binom{n}{3}$ triplets the minimum number of caterpillars required is strictly larger
than the minimum number of trees required.  For example, the set of triplets
$\{1 3 | 4, 1 4 | 2, 1 4 | 3,  2 3 | 1, 2 4 | 1 \}$ requires at least 3 caterpillars but only
2 trees. This can be verified by hand or by a simple adaptation of the ILP.

\section{Open problems}

For each problem 2-$\Pi_i$ that is polynomial-time solveable, we know that all instances
of the problem are ``yes'' instances, which means that $k$-$\Pi_i$
is actually polynomial-time solveable for all $k \geq 2$. An obvious conjecture is
that for each problem 2-$\Pi_i$ that has been shown to be NP-complete in this article,  $k$-$\Pi_i$ is actually NP-complete for all $k \geq 2$.
%We consider it highly likely that this conjecture is true.
However,
generalizing the gadgetry used in this article will probably require a considerable effort and need to go beyond auxiliary computational proofs. In the
same spirit, it seems plausible that $k$-{\sc Tree Compatibility} is hard for every $k \geq 3$.
The fact that $\tau(n) \rightarrow \infty$ as $n \rightarrow \infty$ means that the problem
in any case does not become trivially polynomial-time solveable for a sufficiently large, constant number of trees. Other 
challenges include tightening the bounds on $\tau(n)$ and $\tau_c(n)$ and determining
whether $\tau(n) = \tau_c(n)$ for all $n$. In the applied domain, it will be interesting to explore whether the slow growth of $\tau(n)$ has implications for analysis of incongruent biological datasets.

\bibliographystyle{plain}
\bibliography{2LOarxiv}

%\newpage

%\begin{appendix}

\appendix

\section{An alternative proof of Theorem~\ref{t:2-pi1}}

In this section, we establish an alternative proof that 2-$\Pi_1$ is NP-complete. While the proof is more involved \blue{than} the one presented in Section~\ref{sec:two}, it highlights a link of 2-$\Pi_1$ to the graph-theoretic problem {\sc Dichromatic Number}. As an interesting by-product, we also strengthen the hardness result for {\sc Dichromatic Number} to a subclass of digraphs whose vertices have a bounded out-degree. 

We start by introducing the decision problem {\sc Dichromatic Number} that generalizes the concept of the chromatic number to directed graphs~\cite{neumann1982dichromatic}. More formally, {\sc Dichromatic Number} can be stated as follows.\\

\noindent {\sc Dichromatic Number}

\noindent {\bf Instance.} A digraph $D$ with vertex set $V(D)$ and a positive integer $k$.

\noindent {\bf Question.} Is $D$ {\it $k$-dicolorable}, i.e. does there exist a vertex coloring of $V(D)$ with at most $k$ colors such that the subgraph induced by each color does not contain any directed cycle? \\

It is shown in \cite{bokal2004circular} that {\sc Dichromatic Number} is an NP-complete problem. In particular, the authors showed that the special case of {\sc Dichromatic Number} in which $k$ is fixed to 2 is NP-complete. We refer to this variation of the problem as {\sc 2-Dichromatic Number} and it is exactly this problem that we use as a starting point to show that 2-$\Pi_1$ is NP-complete. 

Let $D$ be the input to an instance $\cI$ of {\sc 2-Dichromatic Number}. Following the notation introduced above, we say that $\cI$ is {\it 2-dicolorable} if $\cI$ is a `yes'-instance and refer to any valid coloring of $D$ with exactly two colors as a {\it 2-dicoloring}. We will see later that a digraph that is 2-dicolorable corresponds to a carefully constructed `yes'-instance of 2-$\Pi_1$.

Let $\cR$ be a set of rooted triplets, and let $X$ be the label set of $\cR$. Let $D(\cR)$ be the digraph whose vertex set is $X$ and for which $(c,a)$ and $(c,b)$ are arcs in $D(\cR)$  precisely if $ab|c$ is an element in $\cR$. We call $D(\cR)$ the {\it triplet digraph} of $\cR$. Furthermore, we say that $\cR$ is {\it caterpillar-compatible} if there exists a caterpillar that displays each triplet in $\cR$.

\begin{lemma} \label{lem:triplet_graph}
 A set $\cR$ of triplets is caterpillar-compatible if and only if the triplet digraph $D(\cR)$ is acyclic.
\end{lemma}

\begin{proof}
Throughout the proof, let $X$ be the label set of $\cR$ with $|X|=n$.

First, suppose that there exists a caterpillar $T$ such that  each triplet $r \in \cR$ is displayed by $T$. Towards a contradiction, assume that there is a  cycle \blue{$C=x_0,x_1,...,x_k,x_0$} in $D(\cR)$. Without loss of generality, we may assume that $C$ is a simple cycle. First note that, by construction of $D(\cR)$, each arc of $C$ corresponds to some triplet in $\cR$ and no two arcs of $C$ correspond to the same triplet. Now, for each arc $(x_i,x_j)$ in $C$  with \blue{$i\in\{0,1,\ldots,k\}$} and \blue{$j=(i+1\mod k)$}, let $r_{i,j}$ be a triplet in $\cR$ that corresponds to $(x_i,x_j)$ in $D(\cR)$. We next show that the triplets \blue{$r_{0,1}, r_{1,2},\ldots, r_{k,0}$} cannot all be displayed by $T$. To see that this is indeed not possible, observe that,  for $T$ to display each triplet $r_{i,j}$ with \blue{$i\in\{0,1,\ldots,k\}$} and \blue{$j=(i+1\mod k)$}, the path from $x_i$ to the root of $T$ is shorter than the path from $x_j$ to the root of $T$. This gives a contradiction and, hence, $D(\cR)$ 
is acyclic.
%But this is easy to see because if they were, triplet $r_{1,2}$ would require taxon 1 to be below taxon 2 in that caterpillar, $r_{2,3}$ taxon 2 to be below taxon 3, ..., $r_{k-1, k}$ taxon $k-1$ to be below taxon $k$. But at the same time triplet $r_{k,1}$ would require taxon $k$ to be below taxon 1 in the same caterpillar, which is clearly impossible. 

Second, suppose that $D(\cR)$ is acyclic. \blue{Since $D(\cR)$ is a directed acyclic graph, it has a topological ordering, say 
$(x_n,x_{n-1},\ldots,x_2,x_1)$, where $x_i$ has no incoming arc in the graph obtained from $D(\cR)$, by deleting the vertices in 
$\{x_n,x_{n-1},\ldots,x_{\steven{i+1}}\}$. Now, let $T$ be a caterpillar on $X$ such that $x_1$ and $x_2$ is the unique cherry of $T$ and, for each $i\in\{3,4,\ldots n\}$, the parent of $x_{i-1}$ is a child of the parent of $x_i$.} Let $x_ix_j|x_k$ be a triplet of $\cR$. By construction of $T$ and because $(x_k,x_i)$ and $(x_k,x_j)$ are arcs in $D(\cR)$, it  follows that $\lca_T(x_i,x_k)=\lca_T(x_j,x_k)<_T\lca_T(x_i,x_j)$. Hence $T$ displays each triplet in $\cR$. 
%We construct a caterpillar $T$ such that each $r \in \cR$ is displayed by $T$. Let $(x_1,x_2,\ldots,x_n)$ be an ordering on the elements in $X$ such that for each $i\in\{1,2,\ldots,n\}$ the vertex $x_i$ has out-degree zero in the digraph obtained from  $D(\cR)$ by deleting each vertex $x_1,x_2,\ldots,x_{i-1}$ and all arcs that are incident with these vertices. Furthermore, let $T$ be a caterpillar on $X$ such that $x_1$ and $x_2$ is the unique cherry of $T$ and, for each $i\in\{3,4,\ldots n\}$, the parent of $x_{i-1}$ is a child of the parent of $x_i$. Now let $x_ix_j|x_k$ be a triplet of $\cR$. By construction of $T$ and because $(x_k,x_i)$ and $(x_k,x_j)$ are arcs in $D(\cR)$, it  follows that $\lca_T(x_i,x_k)=\lca_T(x_j,x_k)<_T\lca_T(x_i,x_j)$. Hence $T$ displays each triplet in $\cR$. 
Combining both cases establishes the lemma.
\end{proof}

% \begin{figure}[h]
%  \centering
%  \includegraphics[width=0.8\textwidth]{./figs/lem_triplet_graph.png}
%  \caption{Here I either I explain this construction using in between step $S'$ and meta-taxa $L_i$ and keep the image, or I simply say like I did now that we take any linear order (caterpillar) induced by the partial order (tree) and remove the image. I'm tempted to do the later, but the image feels useful. (I don't care I spent time making it; if not useful just chop it off.)}
%  \label{fig:tree_to_caterpillar}
% \end{figure}

We next introduce a new variant of the {\sc Dichromatic Number} problem. Let $D$ be a digraph and $v$ be a vertex of $D$. We use $d^+(v)$ to denote the out-degree of $v$ in $D$ and $\Delta^+(D)$ to denote the maximum out-degree of all vertices in $D$. Now consider the following decision problem.\\

\noindent {\sc 2-Dichromatic Number Out-Degree 3} 

\noindent {\bf Instance.} A digraph $D$ with vertex set $V(D)$ and $\Delta^+(D)\leq 3$.

\noindent {\bf Question.} Is $D$ {\it $2$-dicolorable}, i.e. does there exist a vertex coloring of $V(D)$ with at most 2 colors such that the subgraph induced by each color does not contain any directed cycle? \\

The next theorem establishes NP-completeness of {\sc 2-Dichromatic Number Out-Degree 3}.

\begin{theorem}\label{t:2-dicro}
The problem {\sc 2-Dichromatic Number Out-Degree 3} is NP-complete.
\end{theorem}

% For this figure have (a) vertex v, (b) balanced tree on v, (c) gadget on the tree.
\begin{figure}[h] \label{fig:gadget}
 \centering
  \includegraphics[width=0.8\textwidth]{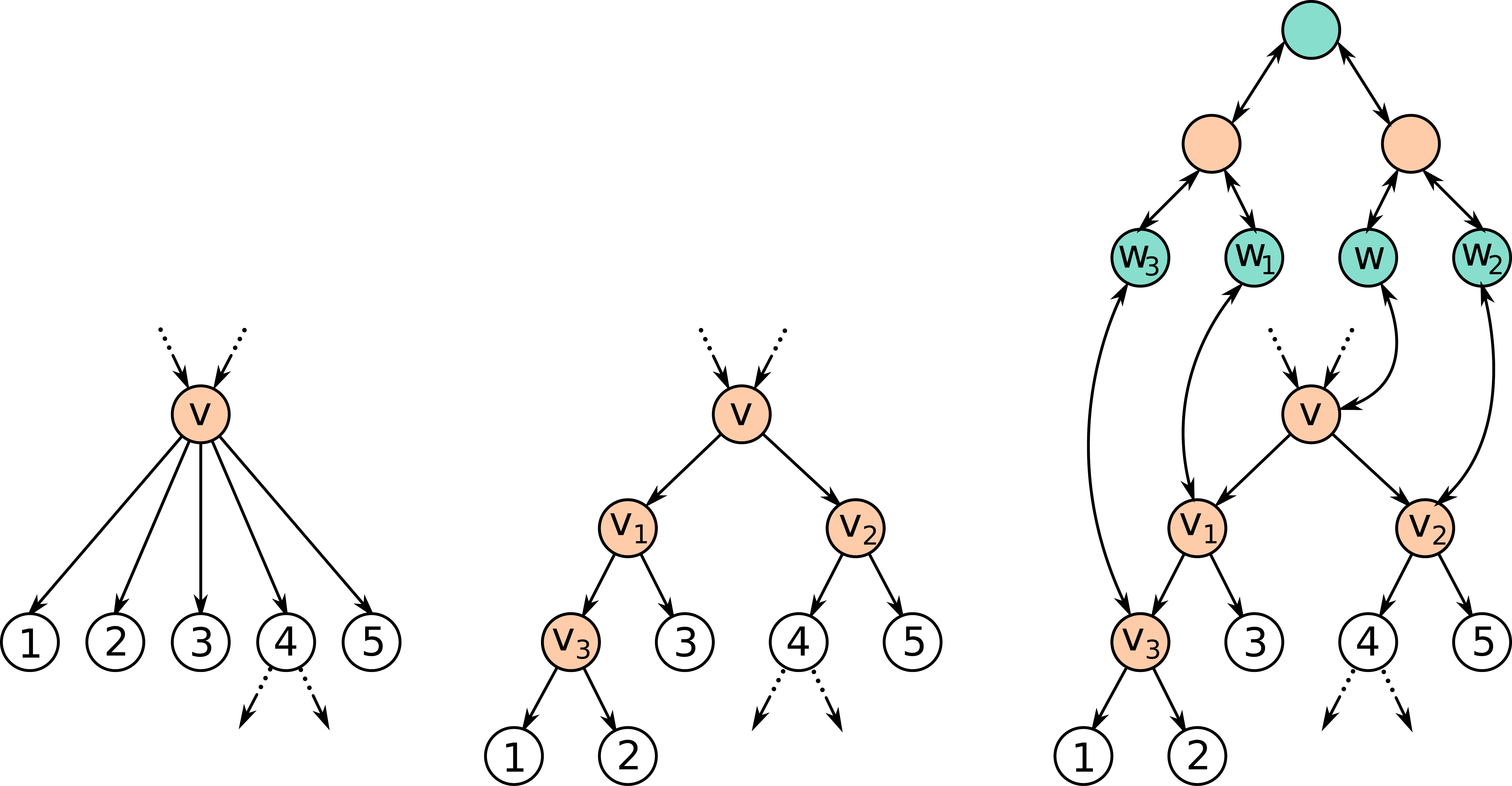}
 \caption{Left: A digraph $D$. Middle and right: A gadget for a vertex $v$ of $D$ whose out-degree is at least 3. For details, see the proof of Theorem~\ref{t:2-dicro}.}
\end{figure}
 %\marginpar {It would be good to add at least one more vertex (or outgoing arcs) to $D$ in the %figure such that not each child of $v$ is a leaf. Also, please change vertex labels $d_i$ to $w_i$ %(see proof).}

\begin{proof}
The problem {\sc 2-Dichromatic Number Out-Degree 3} is clearly in NP since, given a coloring of a digraph it can be verified in polynomial time whether or not each color induces a subgraph that does not contain any directed cycle. We complete the proof, by using a reduction from {\sc 2-Dichromatic Number} to {\sc 2-Dichromatic Number Out-Degree 3}. Let  a digraph $D$ be the input to an instance of {\sc 2-Dichromatic Number}. We construct a digraph $D'$ from $D$ as follows. Start with $D'=D$ and repeat the steps described in the next paragraph for each vertex $v$ in $D$ with $d^+(v)\geq 3$. 

Let $v$ be a vertex of $D'$ that corresponds to a vertex in $D$ with $d=d^+(v)\geq 3$. Replace $v$ with a rooted balanced binary tree $T_v$ on $d$ leaves with $v$ as its root (for an example, see Figure \ref{fig:gadget} (left and middle)). \blue{In other words, when we replace $v$ by $T_v$, we still consider $v$ to correspond to the vertex $v$ in $D$.} Furthermore, the leaves of $T_v$, which may or may not be leaves of the resulting graph, are labeled by the $d$ children of $v$ in $D$. Note that there are, excluding the root which is labeled with $v$, exactly $d-2$ internal vertices which are labeled with $v_1,v_2,\ldots,v_{d-2}$. It is now easy to verify that $d^+(v)=2\leq3$ in $D'$. Next, we add a gadget to $D'$ whose purpose is, as we will see shortly, to force $v$ in $D'$ to have the same color as all its descendants that are contained in $\{v_1,v_2,\ldots,v_{d-2}\}$ in any 2-dicoloring of $D'$. More precisely, we construct a rooted digraph $D_v$ on $d-1$ leaves bijectively labeled with $\{w,w_1,w_2,\ldots,w_{k-2}\}$ that satisfies the following properties: {\bf (i)} $(u,u')$ is an arc in $D_v$ if and only if $(u',u)$ is an arc in $D_v$; {\bf (ii)} each leaf of $D_v$ has the same distance from the root; and {\bf (iii)} $\Delta^+(D_v)\leq 3$. Now connect $D_v$ with $D'$ by adding the following $2(d-1)$ arcs in $$\{(w,v),(v,w),(w_1,v_1),(v_1,w_1),(w_2,v_2),(w_2,v_2),\ldots,(w_{d-2},v_{d-2}),(v_{d-2},w_{d-2})\}.$$ 

It is now easily checked that the digraph $D'$ that ultimately results from the construction described in the last paragraph has $\Delta^+(D')\leq 3$. Furthermore, the construction can be carried out in polynomial time since the size of $D'$ is polynomial in the size of $D$. Lastly, note that each arc $(v,u)$ in $D$ either corresponds to a single arc $(v,u)$ in $D'$ if $d^+(v)< 3$ in $D$ or, corresponds to a path of vertices such that each vertex on that path except for the last is an element in $\{v,v_1,v_2,\ldots,v_{d-2}\}$ if $d^+(v)\geq 3$ in $D$.

The remainder of the proof essentially consists of establishing the following claim.\\

\noindent{\bf Claim.}
 $D$ is 2-dicolorable if and only if $D'$ is 2-dicolorable.\\
 
Throughout the proof of the claim, let $V(D)$ (resp. $V(D')$) be the vertex set of $D$ (resp. $D'$).
 
First, suppose that $D$ is 2-dicolorable. Let $c$ be a 2-dicoloring of $D$. We obtain a coloring of $D'$ in the following way. Assign each vertex $v$ of $D'$ that corresponds to a vertex of $D$ to $c(v)$. Now, for a vertex $v$ of $D$ with $d^+(v)\geq 3$, consider  the rooted digraph $D_v$ as defined in the construction of $D'$. Since $D_v$  satisfies property {\bf (i)}, the only way to color $D_v$ with two colors $\zeta_1$ and $\zeta_2$ is to assign the root to $\zeta_1$ and each other vertex to $\zeta_1$ if it has an even distance to the root of $D_v$ and to $\zeta_2$ otherwise. Without loss of generality, we may assume that each leaf of $D_v$ is assigned to $\zeta_1$. Now consider  the rooted balanced binary tree $T_v$ as defined in the construction of $D'$ and, in particular the set of arcs that join a leaf of $D_v$ to a vertex of $T_v$ (and vice versa). Since $w$ and $v$ form a 2-cycle in $D'$ and, for each $i\in\{1,2,\ldots d-2\}$,  $w_i$ and $v_i$ form a 2-cycle in $D'$, it follows that each vertex in 
$\{v,v_1,v_2,\ldots,v_{d-2}\}$ is assigned to  $\zeta_2$. Up to interchanging $\zeta_1$ and $\zeta_2$, it now follows that this is the unique 2-dicoloring of the subgraph of $D'$ induced by the vertices of $D_v$ and $T_v$. By repeating the described procedure for each vertex in $D$ whose out-degree is at least 3, we obtain a coloring $c'$ of $D'$. We next show that $c'$ is a 2-dicoloring of $D'$. To the contrary, assume that there exists a directed cycle $C'$ in $D'$ whose vertices are all assigned to the same color. Clearly, if $C'$ contains only vertices that are also vertices in $D$ we are done. Therefore, $C'$  contains a vertex that is not a vertex in $D$. Since all vertices of $C'$ are assigned to the same color, no vertex of $C'$ is a vertex of any subgraph $D_v$. Hence, $C'$ only contains vertices of $D$ and, possibly, for a vertex $v$ in $D$ with $d=d^+(v)\geq 3$, the vertex $v$ itself and some vertices in $\{v_1,v_2,\ldots,v_{d-2}\}$. In summary, this implies that $C'$ contains a vertex $v$  with 
$d^+(v)\geq 3$ in $D$. Now, if we contract all vertices in $\{v_1,v_2,\ldots,v_{d-2}\}$ and repeat the procedure for each vertex of $C'$ that corresponds to a vertex of $D$ with out-degree at least 3, the procedure of vertex contractions in $C'$ results in a directed cycle $C$ that only contains vertices of $D$ and, by construction, implies that $c$ assigns all vertices in $C$ to the same color; a contradiction.

Second, suppose that $D'$ is 2-dicolorable.  Since  $V(D)\subseteq V(D')$, we can obtain a coloring $c$ of $D$ from a 2-dicoloring $c'$ of $D'$ by assigning each vertex in $D$ to the same color as in $D'$. Now, towards a contradiction, assume that $c$ is not a 2-dicoloring of $D$. Then, there exists a color $\zeta$ that induces a subgraph of $D$ that contains a directed cycle $C$. By construction of $D'$, it follows that $D'$ contains a directed cycle $C'$ such that the vertex set of $C$ is a subset of the vertex set of $C'$. Moreover,  each arc $(v,u)$ of $C$ either corresponds to an arc $(v,u)$ in $C'$ if  $d^+(v)<3$ in $D$ or corresponds to a directed path from $v$ to $u$ in $C'$ if $d^+(v)\geq 3$ in $D$. Furthermore, since all vertices on the directed path from $u$ to $v$ in $D'$, which consists of at least one arc, have the same color as $v$ (as justified in the last paragraph), it now follows that, under $c'$, $\zeta$ induces a subgraph of $D'$ that contains a directed cycle; a contradiction.
 
The theorem now follows by combining both cases.
\end{proof}

In the remainder of this section, we show that 2-$\Pi_1$ is NP-complete by using a reduction from {\sc 2-Dichromatic Number Out-Degree 3}. Formally, we reduce an instance of {\sc 2-Dichromatic Number Out-Degree 3} to an instance of {\sc 2-Caterpillar Compatibility}. Since 2-{\sc Caterpillar Compatibility} and 2-$\Pi_1$ are essentially equivalent (see Observation~\ref{obs}), the result then follows immediately.

\begin{theorem}
The problem {\em 2}-{\sc Caterpillar Compatibility} is NP-complete. 
\end{theorem}

\begin{proof}
The problem 2-{\sc Caterpillar Compatibility} is clearly in NP because, given two caterpillars $T$ and $T'$, it can be checked in polynomial time if each element of a set of rooted triplets is displayed by $T$ or $T'$. We complete the proof, by using a reduction from {\sc 2-Dichromatic Number Out-Degree 3} to {\sc 2-Caterpillar Compatibility}. Let  a digraph $D$ be the input to an instance of {\sc 2-Dichromatic Number Out-Degree 3}. Let $v$ be a vertex of $D$. Depending on $d^+(v)$, which is at most 3, let $\cR_v$ be one of the followings.
\begin{enumerate}
\item [(i)] If $v$ has exactly one child, say $a$, set $\cR_v=\{ad_v|v\}$,  where $d_v$ does not label any vertex of $D$.
\item [(ii)] If $v$ has exactly two children, say $a$ and $b$, set $\cR_v=\{ab|v\}$.
\item [(iii)] If $v$ has exactly three children, say $a$, $b$, and $c$, set $\cR_v=\{ab|v, ac|v, bc|v\}$.
\end{enumerate}
Now, let $$\cR=\bigcup_{v\in V(D)}\cR_v,$$ where $V(D)$ denotes the vertex set of $D$, be the input to an instance of {\sc 2-Caterpillar Compatibility}. Clearly, the reduction can be carried out in polynomial time and $\cR$ has size polynomial in the size of $V(D)$. Moreover, we make the following crucial observation that we will freely use throughout the rest of the proof. Note that $D(\cR)$, the triplet digraph of $\cR$, is essentially the same graph as $D$. The only difference is that $D(\cR)$  has an extra leaf $d_v$ for each vertex $v$ in $V(D)$ \blue{with} $d^+(v)=1$. However, these extra leaves have no impact upon the acyclicity of $D(\cR)$.

We next establish the following claim.\\

\noindent {\bf Claim.}  $D$ is 2-dicolorable if and only if there exists \blue{two} caterpillars $T$ and $T'$ such that each triplet in $\cR$ is displayed by $T$ or $T'$.\\

First, suppose that $D$ is 2-dicolorable. Let $\zeta_1$ or $\zeta_2$ be the two colors that are used in a 2-dicoloring $c$ of $D$. Furthermore, let $\cR_1$ and $\cR_2$ be the two sets of a bipartition of $\cR$ such that $\cR_1$ contains each triplet of $\cR$ whose witness corresponds to a vertex in $D$ that is assigned to color $\zeta_1$ and, similarly, $\cR_2$ contains each triplet of $\cR$ whose witness corresponds to a vertex in $D$ that is assigned to color $\zeta_2$. Now assume that $\cR_1$ is not caterpillar-compatible. It then follows by Lemma~\ref{lem:triplet_graph} that $D(\cR_1)$ contains a directed cycle $C$. Moreover, the tail of each edge in $C$ corresponds to a witness of some triplet in $\cR_1$. Hence, all vertices of $C$ are assigned to $\zeta_1$; thereby contradicting that $c$ is a 2-dicoloring of $D$.

Second, suppose that there exist two caterpillars $T$ and $T'$ such that each triplet in $\cR$ is displayed by either $T$ or $T'$. Let $\cR_1$ be the subset of triplets of $\cR$ that are displayed by $T$, and let $\cR_2$ be the subset of triplets of $\cR$ that are not displayed by $T$. Observe that each triplet in $\cR$ is encoded by a unique vertex of $D$ and that a vertex $v$ of $D$ is encoded by either a single triplet in $\cR$ (if $d^+(v)<3$) or three triplets in $\cR$ (if $d^+(v)=3$). Now, let $v$ be a vertex of $D$. If $v$ is encoded by a single triplet in $\cR$ that is displayed by $T$ (resp. not displayed by $T)$, then assign $v$ to color $\zeta_1$ (resp. $\zeta_2$). Furthermore, if $v$ is encoded by three triplets in $\cR$, then, by the pigeonhole principle, at least two of those triplets are displayed by $T$ (resp. are not displayed by $T$) in which case we assign $v$ to $\zeta_1$ (resp. $\zeta_2$). It remains to show that the resulting coloring $c$ of $D$ is a 2-dicoloring. Assume the contrary. 
Then there exist a directed cycle \blue{$C=v_0,v_1,\ldots,v_k,v_0$} in $D$ whose vertices are all assigned to the same color under $c$. Without loss of generality, we may assume that this color is $\zeta_1$. Let \blue{$i\in\{0,1,\ldots,k\}$}, and let $j=i+1 \mod k$. We next apply an iterative procedure on each vertex in $C$ and construct a set $\cR_1'$ of triplets. Initially, set $\cR_1'=\emptyset$. Now, repeat the following for each $v_i$ in $C$. If $d^+(v_i)=1$ in $D$, reset $\cR_1'$ to $\cR_1'\cup\{v_jd_{v_i}|v_i\}$. If $d^+(v_i)=2$ in $D$, let $a$ be the child of $v_i$ that is not contained in $C$ and reset $\cR_1'$ to $\cR_1'\cup\{v_ja|v_i\}$. Lastly, if $d^+(v_i)=3$ in $D$, let $a$ and $b$ be the two children of $v_i$ that are not contained in $C$, let $r_1$ and $r_2$ be the two triplets in $\{ab|v_i,av_j|v_i,bv_j|v_i\}$ that are displayed by $T$ and reset $\cR_1'$ to $\cR_1'\cup\{r_1,r_2\}$. In all three cases, it follows by construction of $c$ that, as $v_i$ is assigned to $\zeta_1$, each triplet in 
$\cR_1'$ is displayed by $T$. It is now easily checked that $D(\cR_1')$ contains $C$. Furthermore, by Lemma~\ref{lem:triplet_graph}, we deduce that $\cR_1'$ is incompatible and, hence, $\cR_1$ is also incompatible; thereby contradicting that each triplet in $\cR$ is displayed by $T$ or $T'$.
%If $v_i$ is a vertex of $D$ with $d^+(v_i)<3$, then the arc $(v_i, v_j)$ corresponds to some triplet that is displayed by $T$. Furthermore, if $v_i$ is a vertex of $D$ with $d^+(v_i)=3$, then the arc $(v_i, v_j)$  also corresponds to some triplet that is displayed by $T$. To justify that this is true, note that at least two of the three triplets, say $r_1$ and $r_2$, that are encoded by $v_i$ are displayed by $T$. Since $r_1$ and $r_2$ but those two span all three arcs that leave $v_i$ ????????. Hence, the existence of $C$ in $D$ implies that there is a subset $\cR_1'\subseteq\cR_1$ such that $D(\cR_1')$ contains a directed cycle. It now follows from Lemma~\ref{lem:triplet_graph} that $\cR_1'$ is incompatible and, hence, $\cR_1$ is also incompatible; thereby contradicting that each triplet in $\cR$ is displayed by $T$ or $T'$.

Combining both cases establishes the claim and, therefore the theorem.
\end{proof}

\begin{corollary}
The problem 2-$\Pi_1$ is NP-complete.
\end{corollary}

%\section{3-$\Pi_1$ is NP-hard}
%%\input{./3treesishard/3treesishard.tex}

\section{Auxiliary computational proofs used in the hardness reductions}
In several of the hardness proofs we make use of small instances which, up to symmetry,
are the unique solution to the set of constraints they induce. The symmetries differ
per problem. For $\Pi_5$ (betweenness) and $\Pi_9$ (non-betweenness)
the set of constraints induced by a linear order $\alpha$ is equal to the set of constraints
\steven{induced by $\bar \alpha$, so a linear order is indistinguishable from its reversal.}
For $\Pi_1$ the first two elements of $\alpha$ can be swapped without altering
the set of constraints induced by the order. For $\Pi_6$ no symmetries need to be eliminated.

In each case we located these uniquely defined instances by a brute-force search,
written in Java. The Java code already verifies the uniqueness property itself, but to increase
confidence in the computational proof we then use the constraint programming solver
MiniZinc \steven{(\url{http://www.minizinc.org}, see also \cite{nethercote2007minizinc})} as a second check. In particular, the Java code outputs the discovered linear orders $\{\alpha, \beta\}$ - or in the case of 3-$\Pi_1$, $\{ \alpha, \beta, \gamma \}$ - and the set of constraints
induced by the linear orders. We then ask MiniZinc to look for linear orders that satisfy these
constraints but which differ from $\{ \alpha, \beta \}$ in at least \steven{one position}. MiniZinc
concludes unsatisfiability. If we remove the ``differs
in at least one place'' constraint, and ask MiniZinc to generate all possible solutions, it
\steven{generates} $\{\alpha, \beta\}$ as the only valid solution. The combination of these results
verifies uniqueness. 

One of the results in the paper relies on the existence of linear orders $\{ \alpha, \beta, \gamma\}$ such that these are unique instances not just for 3-$\Pi_1$ but also for $\Pi_1$
on the space of 3 phylogenetic trees. The proof of uniqueness for 3-$\Pi_1$ is described
above. For the variant on the space of 3 phylogenetic trees, we tackle the proof slightly
differently. At a high level, we ask for 3 phylogenetic trees that satisfy all the induced
constraints such that at least one of the trees is not a caterpillar. The goal, again, is to generate
an unsatifiability. To encode the ``at least one of the trees is not a caterpillar'', we demand
that one of the trees contains at least two \emph{cherries}. A cherry is an internal vertex
$u$ such that $u$ has exactly two children, both \blue{leaves}. Every phylogenetic tree on at least
two taxa has at least one \blue{cherry. 
%A phylogenetic tree on two or more taxa is a caterpillar
%if and only if it has exactly one cherry.    
%%%caterpillar definition is given in Section 2
To} search through the space of phylogenetic
trees, we observe that  the $\{ \alpha, \beta, \gamma\}$ in question each have exactly
6 elements. A rooted binary phylogenetic tree on 6 taxa has exactly 4 internal edges. Each
internal edge can be identified by a specific \emph{cluster} i.e. the subset of taxa reachable in a downwards direction from it. The tree is uniquely defined by these 4 clusters. These 4 clusters
form a laminar family (which corresponds to the concept of \emph{compatibility}, in the phylogenetics literature). We use this characterisation to search through tree space, noting that a tree contains two or more cherries if and only if two or more of its 4 defining clusters have
cardinality 2. MiniZinc, as expected, reports unsatisfiability. When the ``at least two cherries'' constraint is removed and MiniZinc is asked to generate all solutions, it generates
$\{ \alpha, \beta, \gamma \}$ as the only solution (as usual, up to elimination of symmetries).

Full source-code for the computational proofs can be downloaded from \url{http://skelk.sdf-eu.org/ternary}.

\section{An Integer Linear Program for computation of $\tau(n)$ and $\tau_c(n)$}

We begin with computation of $\tau(n)$. We describe an ILP that, for a fixed integer $k \geq 1$,
determines whether $\tau(n) \leq k$. The ILP has $3k\binom{n}{3}$ binary variables
$x_{ab|c, t}$ where $t \in \{1, \ldots, k\}$ and $ab|c \in \cT_n$. Variable $x_{ab|c, t}$ is
1 if and only if tree $k$ displays $ab|c$. To ensure that every triplet in $\cT_n$ is displayed
by some tree we add the following covering constraint, where $t$ is assumed to range over $\{1, \ldots, k\}$,
\begin{align}
\sum_{t} x_{ab|c, t} \geq 1, 
\end{align}
for every $ab|c \in \cT_n$.
We also need to add constraints to ensure that we correctly model trees. To do this, we observe firstly that for every subset $\{a,b,c\}$ of 3 distinct leaves a rooted binary tree displays exactly one of the triplets $ab|c, ac|b, bc|a$. Hence, for every tree $t$, and for every subset $\{a,b,c\}$ of 3 distinct leaves, we add the constraint
\begin{align}
x_{ab|c,t} + x_{ac|b,t} + x_{bc|a,t} = 1.
\end{align}
Next we make use of the well-known fact that a compatible set of triplets can be characterised by avoidance of local conflicts. In particular, consider a set of triplets in which, for each subset $\{a,b,c\}$ of 3 distinct leaves, exactly one of the three possible triplets with leaf set $\{a,b,c\}$ is in the set. Then, as proven in \cite{guillemot2013kernel}, this set is compatible (and thus corresponds exactly to the set of triplets displayed by some rooted phylogenetic tree) if and only
if, for every subset of four distinct leaves $\{a,b,c,d\}$, the subset of triplets whose leaves
are in $\{a,b,c,d\}$ are compatible. This, again referring to \cite{guillemot2013kernel}, is equivalent to constraining, for every subset of four distinct leaves $\{a,b,c,d\}$, that
if $ab|c$ and $bc|d$ are in the set, then so are $ab|d$ and $ac|d$. This naturally leads to an ILP with $O(kn^4)$ constraints. Specifically, for every tree $t$ and for every subset $\{a,b,c,d\}$
of 4 distinct leaves, we add the constraints
\begin{align}
x_{ab|c, t} + x_{bc|d, t} - x_{ab|d,t}  &\leq 1\\
x_{ab|c, t} + x_{bc|d, t} - x_{ac|d,t}  &\leq 1.
\end{align}
This concludes the ILP for computation of $\tau(n)$. The ILP for $\tau_c(n)$ is similar. The
only difference is that we add extra constraints which forbid the trees from having
two or more cherries. This can be achieved by adding, for every tree $t$ and for
every ordered tuple $(a,b,c,d)$ of 4 distinct leaves, the constraint
\begin{align}
x_{ab|c, t} + x_{cd|a, t} \leq 1.
\end{align}
%In fact, for computation of $\tau_c(n)$, constraints of type (3) and (4) can actually be omitted,
%although we do not include details.

%\end{appendix}

\end{document}